\documentclass[runningheads]{llncs}

\usepackage{pgfplots}
\pgfplotsset{compat=1.17}

\usepackage{amsfonts}
\usepackage{amssymb}
\usepackage{mathrsfs}
\usepackage{amsmath}
\usepackage{tikz}
\usepackage{graphicx}
\usepackage{subfigure}
\usepackage{mathtools}
\usepackage{amsthm}
\usepackage{graphicx}
\usepackage{bbm}
\usepackage{hyperref}
\usepackage{authblk}
\usepackage{blindtext}
\usepackage{algorithm}
\usepackage{algpseudocode}
\usepackage{csquotes}
\usepackage{subcaption}
\usepackage{float}
\usepackage{authblk}
\usepackage{pgfplots}
\usepackage[T1]{fontenc}
\usepackage{hyperref}
\usepackage{cite}
\usepackage{multirow}
\usepackage{makecell}

\pgfplotsset{compat=1.18}

\usetikzlibrary{calc}
\usetikzlibrary{arrows.meta}
\usetikzlibrary{matrix}

\newcommand{\bb}[1]{\mathbb{#1}}  
\newcommand{\cc}[1]{\mathcal{#1}} 

\newcommand{\nobarfrac}[2]{\genfrac{}{}{0pt}{}{#1}{#2}}

\newcommand{\ct}{\textsf{ct}{}}

\newcommand{\sk}{\textsf{sk}{}}
\newcommand{\rlk}{\textsf{rlk}{}}
\newcommand{\rtk}{\textsf{rtk}{}}
\newcommand{\evk}{\textsf{evk}}
\newcommand{\pk}{\textsf{pk}{}} 
\newcommand{\eval}{\textsf{Eval}}
\newcommand{\REP}{\textsf{REP}{}}
\newcommand{\PE}{\textsf{PE}{}}
\newcommand{\REQ}{\textsf{ReQ}}

\newcommand{\dec}{\textsf{Dec}}
\newcommand{\enc}{\textsf{Enc}}

\begin{document}
\title{Cryptanalysis on Lightweight Verifiable Homomorphic Encryption}

\author{Jung Hee Cheon\inst{1,2}, Daehyun Jang\inst{1}}
\authorrunning{Cheon, Jang}

     \institute{
    Seoul National University, Seoul, Republic of Korea \\ 
         \email{
         \{jhcheon, jadh0309\}@snu.ac.kr
         }
         \and
         CryptoLab Inc., Seoul, Republic of Korea
     }

\maketitle 

\begin{abstract}{
    Verifiable Homomorphic Encryption (VHE) is a cryptographic technique that integrates Homomorphic Encryption (HE) with Verifiable Computation (VC). It serves as a crucial technology for ensuring both privacy and integrity in outsourced computation, where a client sends input ciphertexts $\ct$ and a function $f$ to a server and verifies the correctness of the evaluation upon receiving the evaluation result $f(\ct)$ from the server.
        
    Chatel et al.~\cite{VE} (CCS'24) introduced two VHE schemes: Replication Encoding (\REP) and Polynomial Encoding (\PE). A similar approach to \REP{} was used by Albrecht et al.~\cite{cryptoeprint:PRF} (EUROCRYPT'24) to develop a Verifiable Oblivious PRF scheme (vADDG).   
    A key approach in these schemes is to embed specific secret information within ciphertexts and use them to verify homomorphic evaluations.
    
    This paper presents efficient forgery attacks against the verifiability guarantees of these VHE schemes. We introduce two attack strategies. The first targets $\REP$ and vADDG, extracting secret information in encrypted form from input ciphertexts and leveraging it to forge output ciphertexts without being detected by the verification algorithm. The second targets $\PE$, exploiting its secret embedding structure to forge output ciphertexts that remain valid on input values for verification, yet violate the verifiability property.
     
     Our forgery attack on vADDG demonstrates that the proposed 80-bit security parameters provide at most 10 bits of concrete security. Our attack on \REP{} and \PE{} achieves a probability 1 attack with linear time complexity when using fully homomorphic encryption.
    }
\end{abstract}


\section{Introduction}
Verifiable Computation (VC) is a technique that guarantees the correctness of the result of a computation outsourced by a client~\cite{First_VC}. This technique allows the client to detect and prevent erroneous or malicious computations by the server. This ensures the integrity of the computations. On the other hand, Homomorphic Encryption (HE) allows computations to be performed on encrypted data without having to decrypt. HE allows the privacy of sensitive information to be maintained while complex calculations are outsourced to the server. Overall, in outsourced computing, VC ensures the integrity of the computational results, while HE guarantees the privacy of the client's data. Verifiable Homomorphic Encryption (VHE) is the combination of these two technologies.
A naive approach to VHE is the use of the SNARK or other VC techniques for homomorphic computation on encrypted data~\cite{SnarkVHE}. However, this approach results in a highly inefficient solution due to several non-arithmetic operations required in HE. These operations do not integrate seamlessly with existing VC techniques. 

Alternatively, other lines of research~\cite{VE, cryptoeprint:PRF} have explored the design of efficient verifiable homomorphic encryption schemes without combining VC techniques with HE schemes.
These schemes aim to achieve VHE at lower computational cost by concealing specific secret information within ciphertexts and precomputed values, and using them to ensure verifiability.
In this paper, we review such \textit{`lightweight'} verifiable homomorphic encryption schemes and demonstrate that these approaches have been insufficient, highlighting the need for further investigation and improvement.

    \subsection{Lightweight Verifiable Homomorphic Encryptions}
        
        \subsubsection{vADDG.}
        An Oblivious Pseudorandom Function (OPRF) is a two-party protocol between a server~$\mathcal{A}$ holding the secret PRF key~$\bf k$, and a client~$\mathcal{B}$ holding its private input~$\bf x$. At the end of the protocol, $\mathcal{B}$ obtains $F_{\bf k}({\bf x})$, while $\mathcal{A}$ learns nothing about $x$ and $\mathcal{B}$ learns nothing about $\bf k$.
        HE provides a direct solution for constructing an OPRF: 
        $\mathcal{B}$ encrypts the inputs under an HE scheme and sends the ciphertexts to $\mathcal{A}$. 
        Then, $\mathcal{A}$ homomorphically evaluates the PRF using its private key and returns the result to $\mathcal{B}$. 
        Finally, $\mathcal{B}$ decrypts the received ciphertext to obtain the output.
        
        Albrecht {\em et al.}~\cite{cryptoeprint:PRF} leveraged this idea to propose a new candidate for an OPRF, referred to as the ADDG scheme, based on the TFHE homomorphic encryption scheme~\cite{CGGI}. 
        They further extended it to a Verifiable OPRF (VOPRF), which we refer to as the vADDG scheme.
        For the verifiability extension, vADDG adopts the following method: First, $\mathcal{A}$ publishes a set of verification values (i.e., challenges) $\{\mathbf{x}^\star_k\}_{1 \leq k \leq \kappa}$ together with their PRF evaluations $\{\mathbf{z}^\star_k\}_{1 \leq k \leq \kappa}$ and corresponding zero-knowledge proofs.  
        Suppose that $\mathcal{B}$ wishes to evaluate the OPRF $F_{\mathbf{k}}$ on private inputs $\mathbf{x}_i$ ($1 \leq i \leq \alpha$) with verification. 
        $\mathcal{B}$ forms a tuple consisting of $\nu$ copies of each $\mathbf{x}_i$ and $\beta$ verification values from the pre-published set $\{\mathbf{x}^\star_k\}_{1 \leq k \leq \kappa}$. 
        This tuple, of length $\alpha\nu+\beta$, is then randomly permuted, encrypted, and finally sent to $\mathcal{A}$ and computed.
        
        To verify the integrity of the result after outsourcing, \(\mathcal{B}\) first decrypts and recovers the random permutation. Then $\cc B$ checks whether the evaluation of \(\mathbf{x}^\star_k\)'s matches that of \(\mathbf{z}^\star_k\)'s. It also checks that the copies of each \(\mathbf{z}_i\) all have the same value. Then $\cc B$ accepts the \(\mathbf{z}_i\)'s as honest results. 
        A notable feature of this (v)ADDG scheme is that the homomorphic computation of this PRF circuit requires only one level of bootstrapping depth by removing the key-switching key.

        \subsubsection{VERITAS.} 
        Chatel {\em et al.}~\cite{VE} proposed a novel VHE scheme, called {\sf Veritas}, which supports all operations in BFV and BGV. Compared to the baseline HE schemes, {\sf Veritas} introduces an overhead ranging from a factor of 1 to a two-digit factor, depending on the characteristics of the circuit. The main idea of these schemes is similar to that of vADDG scheme. For a given circuit $f$, the client holds random verification values $v_i$'s and their precomputed results $f(v_i)$'s. The client then encrypts a message along with the verification values using specially designed encoders, namely Replication Encoding (\REP) and Polynomial Encoding (\PE). If the homomorphic computation is correct, the decrypted result will contain both the precomputed results and the desired computational result. To verify the integrity of the result, the client checks whether the computed values match the previously known values. If they do, along with additional verifications, the client accepts the result as valid.

        \paragraph{Replication Encoding.} \REP{} encodes given messages and verification values as follows: Let \( n \) be a power of two. Among the slots indexed from \( 1 \) to \( n \), the verification values \( v_i \)'s are placed in certain slots indexed by \( S \subset \{1, \dots, n\} \) with \( |S| = n/2 \). Meanwhile, the other slots indexed by \( S^c = \{1, \dots, n\} \setminus S \) are repeatedly filled with the message \( m \). For example, if \( n = 4 \) and \( S = \{1, 4\} \), the client creates a vector \( (v_1, m, m, v_2) \) for a message \( m \) and challenge values \( v_1, v_2 \). We refer to the former as the verification slots and the latter as the computation slots. If the computation is correctly performed, the resulting vector will be \( (f(v_1), f(m), f(m), f(v_2)) \). The client verifies as follows: Check that all values in the verification slots match the precomputed values \( f(v_i) \)'s and that the values in the computation slots are identical. If both conditions are satisfied, the client accepts the computation result as \( f(m) \).

        \paragraph{Polynomial Encoding and ReQuadratization.} 
        \PE{}, the other encoding method in \cite{VE}, encodes a message \( \mathbf{m}\in \mathbb{Z}_t^N \) as follows: the client randomly chooses \( \alpha \gets \mathbb{Z}_t^\times \) and a verification value \( \mathbf{v} \gets \mathbb{Z}_t^N \). Next, the client interpolates \( \mathbf{m} ,\mathbf{v}\) at \( Y = 0 \) and \( Y = \alpha \), respectively. As a result, the client obtains a linear polynomial
        \( \mathbf{m} + \left(\frac{\mathbf{v} - \mathbf{m}}{\alpha}\right) Y \in \mathbb{Z}_t^N[Y] \). Then, this polynomial is coefficientwise encrypted to a ciphertext polynomial $\ct_0 + \ct_1 Y \in \mathcal{R}_q^2[Y]$. Here, $\mathcal{R}_q^2$ is the BFV ciphertext space and $\mathcal{R}_q^2[Y]$ is the set of polynomials in the variable $Y$ with coefficients in $\mathcal{R}_q^2$, and the operations on the coefficients correspond to homomorphic operations. For verification, the client checks whether \( \dec(\mathsf{F}(\alpha)) = f(\mathbf{v}) \), and if they match, the constant term \( \dec(\mathsf{F}(0)) \) is accepted as the desired computation result \( f(\mathbf{m}) \).

        However, the computation cost of \PE{} increases exponentially with the multiplication depth: squaring \(\ct_0 + \ct_1Y\) yields \(\ct'_0 + \ct'_1Y + \ct'_2Y^2\), and squaring it again results in \(\ct''_0 + \dots + \ct''_4Y^4\). This results in an exponential performance degradation. To address this issue, \cite{VE} proposed the ReQuadratization (\REQ{}), a client-aided protocol that transforms a quartic ciphertext polynomial into a quadratic ciphertext polynomial. See Section V.D and Appendix E of \cite{VE} for more details.

    \subsection{Attack Description on Lightweight VHE Schemes}
    A fundamental requirement of a VHE scheme is that the server must homomorphically evaluate the designated function $f$ on the client’s ciphertexts. 
    If the server can instead produce a valid ciphertext corresponding to another function $g \neq f$ that still passes the verification, the verifiability guarantee of the scheme is broken. 
    In our forgery attacks on lightweight VHE schemes, the server does not need to know the embedded secret information in plaintext. 
    Rather, these attacks leverage the capabilities of the HE schemes to obliviously recover the secret information in its \emph{encrypted form}, and then use it to forge valid ciphertexts.
    
        \subsubsection{Attack on vADDG.}
        For the vADDG scheme, the essential step of the forgery attack is to homomorphically recover the positions in the encrypted state where identical values appear.
        However, since the ADDG scheme utilizes the TFHE scheme which supports only a single level of bootstrapping depth, it is difficult to carry out the forgery attack using the published verification values. To address this, we propose a method for extracting positions by evaluating a characteristic function. Specifically, the adversary can construct a characteristic function that takes each TFHE ciphertext string as input and outputs an encryption of \(0\) or \(1\) and then adds that output to the ciphertext string itself. The forgery attack succeeds if this characteristic function outputs \(0\) at the verification value and outputs \(1\) at one of the computation values.
        Using this approach, we obtain a circuit that achieves a forgery probability of roughly $2^{-10}$, showing that the concrete security of the 80-bit parameter set in~\cite{cryptoeprint:PRF} is significantly weaker in practice.

        \subsubsection{Attack on Replication Encoding.}
        Similar to the previous forgery attack on vADDG, in \REP{} the server can forge the ciphertexts without knowing the verification slot, but only with their information in the encrypted state by utilizing homomorphic computation. The essential step of this attack is to identify the computation slots \(S^c\) in the encrypted state.

        The adversary can evaluate a cheating circuit to extract the positions of common values among the $n$ slots. 
        For instance, when $n=4$, the circuit
        \[
        (v_1, m, m, v_2) \mapsto (0, 1, 1, 0)
        \]
        enables the adversary to identify the verification slots in the encrypted state. 
        We design a simple circuit that extracts this information from a fresh ciphertext using homomorphic comparison over $\mathbb{Z}_t$. 
        After recovering the verification slots in the encrypted state, the adversary can forge the computational result. 
        For example, the following vector
        \[
        (f(v_1), f(m), f(m), f(v_2)) \odot (1, 0, 0, 1) \;+\; (g(v_1), g(m), g(m), g(v_2)) \odot (0, 1, 1, 0)
        \]
        places the verification values $f(v_i)$ in the verification slots and the malicious values $g(m)$ in the computation slots, where $g$ is any circuit different from $f$, possibly malicious. 

        \subsubsection{Attack on Replication Encoding with Multiple Secret Keys.}
        The above attack circuit utilizes homomorphic rotations to identify the positions of repeated values across different slots. 
        A straightforward countermeasure is to prevent such rotations by withholding the rotation key of index~1 from the server; however, this not only disables bootstrapping—often unacceptable in practice—but also depends on the heuristic assumption that the rotation key of index~1 cannot be derived from other materials.
        Instead, the attack can be mitigated by employing \emph{multiple secret keys}, thereby preventing computations across different slots. 
        This method allows bootstrapping to remain available, while simultaneously blocking the adversary from carrying out the previous attack. 
        We denote this variant of \REP{} as $\REP^{\sf msk}$.

        However, $\REP^{\sf msk}$ can also be attacked in the same manner as the attack on vADDG.
        By evaluating a characteristic function on each slot for with a randomly chosen support $A$, if a value in a slot belongs to the set \(A \subset \mathbb{Z}_t\), then after evaluation that slot will contain 1; otherwise, it will include 0. Therefore, if \(A\) contains only the message and not the verification value, this forgery attack will succeed.
        We present a pseudorandom characteristic function with a cost of \(O(\log t)\) homomorphic multiplication. When the adversary implements this characteristic function, the attack success probability becomes greater than \((e(1+n/2))^{-1}\), which is non-negligible. In contrast to the cryptanalysis on vADDG, where security could be ensured by adjusting parameters, in this case, an attack success probability of \(O(n^{-1})\) always occurs regardless of the parameter choice.

        \subsubsection{Attack on Polynomial Encoding.}
        We now present attacks on \PE{} assuming access to the \REQ{} protocol, which enables the server to perform a circuit with large multiplicative depth. In the previous attack, the verification slots \(S\), the secret of $\REP$, were recovered in an encrypted form and exploited in a forgery attack. In \PE{}, $\alpha$ is the secret, and similarly, if $\alpha$ is recovered in encrypted form, a forgery attack can be mounted. However, it is difficult to obtain a ciphertext $\ct$ containing information about $\alpha$ through homomorphic computation, since $\alpha$ is only used for interpolating random values and hence even the plaintexts reveal nothing about $\alpha$.
        \footnote{However, one might regard \(\enc(0,\dots,0) + \enc(1,\dots,1)Y \in \mathcal{R}^2_q[Y]\) as an encryption of \(\alpha\). See Discussion in Section \ref{sec: Discussion}.}.
    
        Nevertheless, unlike the previous case, there is another way to forge it: There exists a vulnerability in the encoding structure. Since the encoding method of \PE{} is an interpolation at \(Y=0\) and \(Y=\alpha\), by leveraging the algebraic properties of \(\alpha\in \bb Z_t^\times\), it is possible to interpolate the desired computed ciphertexts at \(Y=0\) and \(Y=\alpha\) without obtaining any information about \(\alpha\) or its encryption.
        
        Let \(\mathsf{F}(Y) \in \mathcal{R}^2_q[Y]\) denote the honest output requested by the client, and let \(\mathsf{G}(Y) \in \mathcal{R}^2_q[Y]\) denote any malicious output. We define \(\mathsf{H}(Y)\) as follows:
            \[
            \mathsf{H}(Y) = \mathsf{G}(Y) + \Bigl(\mathsf{F}(Y) - \mathsf{G}(Y)\Bigr) \cdot Y^{\phi(t)}.
            \]
        We design a polynomial circuit that operates on \(\mathcal{R}^2_q[Y]\) and deterministically outputs \(\mathsf{H}(Y)\) (up to the \REQ{} protocol) in $O(\log t)$. Consequently, the output \(\mathsf{H}(Y)\) will pass the client's verification, that is, \(\dec(\mathsf{H}(\alpha)) = \dec(\mathsf{F}(\alpha)) = f(\mathbf{v})\), while \(\mathsf{H}(0) = \mathsf{G}(0)\) is decrypted to yield a forged result.

    \subsection{Methodology: Homomorphic Cryptanalysis}
    The attacks presented in this paper rely on neither deep mathematics nor sophisticated arguments, but instead maximize the potential power of homomorphic computation to exploit the weakness of the schemes.
    It is valuable to examine our methodology in a more abstract way, as it may offer further applications in the design of VHE.
    
    The fundamental approach to achieving verifiability in these lightweight schemes is to hide precomputed values via additional secrets alongside the secret key of HE. Both vADDG and \REP{} adopt an index set $S$ as a secret value, which determines the slot or ciphertext position. Similarly, in the case of \PE{}, the element $\alpha$ from \( \mathbb{Z}_t^\times \) serves as the secret value. If an adversary obtains knowledge of these secret values, i.e., $S$ and $\alpha$, a forgery attack becomes trivial by manipulating the secret encryption. Therefore, these schemes conceal the secret values within ciphertexts.
    
    Our approach to attack the schemes is to obtain the secret values {\em encrypted states} and use them to generate a cheating circuit.
    For vADDG and \REP{}, the central idea of the attack is to homomorphically test whether two slots encrypt the same message.
    This enables the adversary to produce an encryption of 1 at positions corresponding to equal plaintexts via homomorphic operations. 
    In contrast, for \PE{}, the security hinges on the secrecy of $\alpha$, as verification is carried out by evaluating \( Y \) at \( \alpha \). 
    While the ReQuadratization employs a re-randomization procedure to obscure $\alpha$ and prevent its computation, our attacks bypass this defense by analyzing the encoding structure of \PE{} and devising a tailored circuit based on interpolation at \( Y=0 \) and \( Y=\alpha \).
    
    In any case, even without knowing the secret value itself, the adversary can homomorphically implement a cheating circuit by leveraging the functionalities of homomorphic encryption. 

    \subsection{Related Works}
    There are two major security concerns in information security: integrity and confidentiality. As data utilization expands beyond mere storage to computational applications, these concerns have evolved into the notions of Verifiable Computation and Homomorphic Encryption, respectively.
    To achieve VC, various cryptographic primitives have been proposed, including SNARKs~\cite{SoK}, Homomorphic MACs~\cite{HomMAc1, Hommac2}, and Homomorphic Signatures~\cite{HomomorphicSignautre}. Meanwhile, since Gentry’s groundbreaking discovery~\cite{Gen09}, HE has seen continuous advancements in both functionality and efficiency. Notable schemes such as BGV, BFV, TFHE, and CKKS~\cite{BGV, bfv, CGGI, CKKS} have significantly improved HE’s practical applicability.
    
    A natural follow-up question is whether VC and HE can be integrated, allowing both concerns to be addressed simultaneously.
    There are two major approaches to constructing Verifiable HE~\cite{FHESNARKSNARKFHE}. 
    One approach is to apply integrity mechanisms to the ciphertext operations of HE~\cite{SnarkVHE, Flexible, Efficiently, BoostingVCE, Rinocchio}. 
    Another approach is to achieve verifiability by generating SNARK-based proofs in the encrypted domain~\cite{heliopolis, HowtoProveStatementsObliviously, BlindzkSNARKsvCOED}. 
    However, adding verifiability to HE is a non-trivial task. 
    Both approaches suffer from significant efficiency limitations. Compared to native HE computations, this overhead is larger by four to five orders of magnitude.

    For the former, conventional integrity mechanisms such as SNARKs are incompatible with the non-arithmetic ciphertext operations of HE. As a result, these operations must be emulated by arithmetic operations, which incurs significant overhead when proving them.
    For the latter approach, the plaintext space must be suitable for generating proofs, which leads to inefficiency in parameter selection, and the process of homomorphically evaluating proofs introduces substantial overhead.

    The verifiable homomorphic encryption scheme targeted in this paper follows a lightweight approach different from the two major existing methods. The idea is to achieve verifiability by hiding verification secrets under the semantic security of HE and by altering the message encoding of the underlying scheme accordingly~\cite{VE, cryptoeprint:PRF, DataSeal}.
    However, there is no well-established hardness assumption or security proof that justifies a reduction from the confidentiality of HE to its integrity, making the connection between the two notions unclear.
    
    Our attack leverages the homomorphic properties of HE to forge the underlying encoding structure. This idea was implicitly suggested in~\cite{CHLR18}, and later formalized in the construction of vADDG~\cite{cryptoeprint:PRF}, where the authors proposed a naive hardness assumption that constructing such a cheating circuit within shallow depth is infeasible.
    This suggests the need for a deeper exploration of what kinds of computations are possible in various homomorphic encryption settings.


\section{Preliminaries} \label{sec 2: Preliminary}
In this section, we present the formal definitions and necessary background required for this paper. In particular, regarding the formal definition and security of the verifiable homomorphic encryption, we follow the definition from~\cite{vFHE} with the necessary modifications.
    \subsection{Homomorphic Encryption}
    Let $\cc M$ be a plaintext space and let $\cc F \subset \{f_\alpha \mid f_\alpha : \cc M^{n_\alpha} \to \cc M\}$ be a set of functions on tuples of $\cc M$.
    Then Homomorphic Encryption (over $\cc F$) is stated as follows.
    
    \begin{definition}[Homomorphic Encryption]
            A Homomorphic Encryption over $\cc F$ is a tuple of PPT algorithms 
            \[
            \Pi = (\mathsf{KeyGen, Enc,Eval,Dec})
            \]
            satisfying the following:
            \begin{itemize}
                \item $(\pk,\sk)\gets {\sf KeyGen}(1^\lambda,\cc F)$: $\pk$ includes evaluation keys $\evk$.
                \item $\ct_x\gets {\sf Enc_{pk}}(x)$ for $x\in \cc M$
                \item $\ct_y\gets{\sf Eval}_{f,\evk}(\ct_x)$ for $y=f(x)$ where $f\in \cc F.$
                \item $y\gets{\sf Dec}_{\sk}(\ct_y).$
            \end{itemize}
    \end{definition}
    \begin{definition}[$\delta$-Correcteness]\label{def:HE_correctness}
        Let $\Pi$ be a Homomorphic Encryption scheme. $\Pi$ is \emph{$\delta$-Correct} if the evaluation on the ciphertext decrypts to the correct result. For all $f$ and $x$, it should be satisfied that
        \[\Pr\left[{\sf Dec}_{\sf sk}(\ct_y)=f(x):                    \begin{aligned}
                    (\pk, \sk) &\gets \mathsf{KeyGen}(1^\lambda,\cc F)\\
                    \ct_x &\gets {\sf Enc}_{\mathsf{\pk}}(x)\\
                    \ct_y &\gets {\sf Eval}_{f, \evk}(\ct_x)
                \end{aligned}\right]\geq 1-\delta.\]
    \end{definition}

    \subsubsection{Bootstrapping.}
    In many homomorphic encryption schemes, errors accumulate with each operation. Bootstrapping is a technique used to remove these errors and refresh the ciphertext. If computations can be performed without limit, the scheme is called Fully Homomorphic Encryption (FHE); otherwise, it is known as Somewhat Homomorphic Encryption (SHE). In many cases, bootstrapping transforms SHE into FHE.
    
    \subsubsection{TFHE.}
    TFHE~\cite{CGGI} is a FHE scheme that supports gate-level bootstrapping, that the (v)ADDG scheme is instantiated with.
    Its plaintext space is $\bb Z_Q$ and the message space $\bb Z_P$ is encoded into $\bb Z_Q$ for $P<Q$. Thereafter a plaintext is encrypted into the ciphertext space $\cc C=\bb Z_Q^{n+1}$ where $n$ is determined by the security parameter. 
    
    TFHE bootstrapping requires a bootstrapping key $\sf{btk}$ for evaluation, which is an encryption of the secret key $\sk$ under a different secret key $\sk'$. This bootstrapping key transforms a ciphertext under $\sk$ into a ciphertext under $\sk'$, and a ciphertext under $\sk'$ can be switched back to one under $\sk$ by using the key-switching key $\sf{ksk}$. The TFHE bootstrapping has an additional feature beyond error refreshing. A lookup table evaluation can be implemented during the bootstrapping process. This bootstrapping technique is called Programmable Bootstrapping (PBS) or functional bootstrapping. PBS enables the homomorphic evaluation of non-linear functions in TFHE. For more details, see \cite{Joy21}.
    
    \subsubsection{BFV.}
BFV is an RLWE-based HE scheme that supports the computation of vectors of integers~\cite{Bra, bfv}, 
and it is the scheme in which \textsc{Veritas} is instantiated. 
Its message space is \( \mathbb{Z}_t^N \), where $N$ is a power of $2$ and \( t=p^r \) is a power of a prime \( p \). 
For \( q>t \), the ciphertext space is 
\(
    \mathcal{R}^2_q = \bigl(\mathbb{Z}[X]/(q,X^M+1)\bigr)^2
\)
where $M$ a power of $2$ and $N=M/d$ where $d$ is a multiplicative order of $p\in \bb Z_M^\times$.
Multiplication is given by the Hadamard product, i.e., element-wise multiplication. 

A notable feature of BFV is that it supports rotations 
\( \sigma_i:\mathbb{Z}_t^N \to \mathbb{Z}_t^N \), 
which shift vectors in \( \mathbb{Z}_t^N \) with the aid of the rotation key \( \rtk_i \). 
In addition, ciphertext–ciphertext multiplication requires the relinearization key \( \rlk \). 
As far as known, BFV bootstrapping requires the evaluation keys $\rtk_1$ and $\rlk$. 
For further technical details, refer to Appendix~\ref{Appendix BFV}.
        
    \subsection{Verifiable Homomorphic Encryption}\label{subsec:2.2}
    Verifiable Homomorphic Encryption is stated as follows. Many works in the literature introduce a tag as part of the output of $\enc$ and $\eval$, but we will treat it as part of the ciphertext, when it exists. Let $\cc M$ be a plaintext space and let $\cc F \subset \{f_\alpha \mid f_\alpha : \cc M^{n_\alpha} \to \cc M\}$ be a set of functions on tuples of $\cc M$.
        \begin{definition}
            A Verifiable Homomorphic Encryption over $\cc F$ is a tuple of PPT algorithms 
            \[
            \Pi = (\mathsf{KeyGen, Enc,Eval,Verify,Dec})
            \]
            satisfying the following:
            \begin{itemize}
                \item $\pk,\sk\gets {\sf KeyGen}(1^\lambda,\cc F)$: 
                \item $\ct_x\gets {\sf Enc_{pk}}(x)$ for $x\in \cc M$.
                \item $\ct_y\gets{\sf Eval}_{f,\evk}(\ct_x)$ for $y=f(x)$ where $f\in \cc F$.
                \item $b\gets {\sf Verify_\sk}(\ct_y)$, the client accepts if $b=0$ and rejects if $b=1$.
                \item $y\gets{\sf Dec}_{\sk}(\ct_y).$
            \end{itemize}
        \end{definition}
        
        \begin{definition}[$\delta$-Completeness]\label{def:VHE_Completness}
            A Verifiable Homomorphic Encryption scheme $\Pi$ is \emph{$\delta$-complete} if the client always accepts a correct output. For all $f$ and $x$, it should be satisfied that
            \[
                \Pr\left[
                0\gets{\sf Verify}_{\sf sk}(\ct_y):                    
                \begin{aligned}
                        \pk, \sk &\gets \mathsf{KeyGen}(1^\lambda,\cc F))\\
                        \ct_x&\gets {\sf Enc}_{\mathsf{pk}}(x)\\
                        \ct_y&\gets {\sf Eval}_{f, \evk}(\ct_x)
                \end{aligned}
                \right]\geq1-\delta
            \]
        \end{definition}
        
        \subsubsection{Malicious Adversary Model.} We now formally define the security of VHE. We assume that the underlying scheme $\Pi$ satisfies both $\delta$-correctness and $\delta$-completeness for some negligible $\delta$.
        Throughout this work, we consider two types of adversaries: malicious and covert. 
        The following definition for malicious adversaries is a modification of the corresponding definition in~\cite{vFHE}. 
        Specifically, the modification is the removal of the tag, as noted above.

        \begin{definition}[Security against Malicious Adversary, modified from \cite{vFHE}]\label{def:VHE_Soundness_Malicious}
        A Verifiable Homomorphic Encryption scheme $\Pi$ is \emph{sound in the presence of malicious adversaries} if the client rejects an incorrect output with overwhelming probability in $\lambda$ for any malicious PPT adversary $\cc A_{\sf mal}$. For all $f$, it should be satisfied that
            \[
                \Pr\left[
                \begin{aligned}
                    0\gets{\sf Verify}_{\sf sk}(\ct_y)\\
                    {\sf Dec}_{\sk}(\ct_y)\neq f(x)
                \end{aligned}
                :                    
                \begin{aligned}
                        \pk, \sk &\gets \mathsf{KeyGen}(1^\lambda,\cc F))\\
                        x& \gets \cc A^{\evk}_{\sf mal}\\
                        \ct_x &\gets {\sf Enc}_{\mathsf{pk}}(x)\\
                        \ct_y &\gets \cc A^{\evk}_{\sf mal}(\ct_x)\\
                \end{aligned}
                \right]\text{ is negligible.}
            \]         
        \end{definition}
        
        \subsubsection{Covert Adversary.}
        Originally introduced in the context of secure multiparty computation in~\cite{Covert}, the covert adversary model can be extended to the VHE setting, given that HE can be regarded as a special case of two-party computation, and was adopted as the adversary model in~\cite{VE}.
        The $\varepsilon$-covert adversary model, assumes that any adversarial deviation is detected with probability at least $\varepsilon$. A customary choice of the deterrence factor $\varepsilon$ is a noticeable probability, e.g., $1/{\sf poly}(\lambda)$. Unless otherwise stated, we use the term \emph{covert adversary} to refer to a $1/{\sf poly}(\lambda)$-covert adversary.
        In this paper, we provide an informal definition of the covert adversary model. 
        For a formal treatment, we refer the reader to Definition~3.4 from~\cite{Covert}. 

        \begin{definition}[Security against Covert Adversaries, informal]
        Let $\varepsilon(\lambda) \in [0,1]$. A Verifiable Homomorphic Encryption scheme $\Pi$ is \emph{sound in the presence of $\varepsilon$-covert adversaries} if for any PPT adversary $\mathcal{A}$, any deviation from the prescribed protocol is detected with probability at least $\varepsilon$.
        \end{definition}
        
        \subsubsection{Access to Oracles.} We assume that the adversary can access to the evaluation keys, but not to the decryption, verification, or encryption oracles. In the attack on \REP{} and \PE{}, we always assume that the adversary has $\rlk$, and in the attack on vADDG, we assume that the adversary has $\sf btk$. The only exception is that in the case of \PE{}, we assume that the adversary can access the \REQ{} protocol.
        We denote adversary models as follows: for example, $\cc A_{\sf mal}^{\rtk_1,\cc O_{\REQ}}$ denotes a malicious adversary who can access $\rtk_1$, $\rlk$, and $\REQ$, while $\cc A_{\sf cov}$ denotes an $1/{\sf poly}(\lambda)$-covert adversary who can access only $\rlk$. (Since we always assume that all adversaries can access $\rlk$, we omit it from the notation.)
        

\section{Cryptanalysis of vADDG} \label{subsec:SHE}
    \subsection{(v)ADDG Scheme}
    In \cite{cryptoeprint:PRF}, Albrecht et al. presented an Oblivious PRF candidate and its extension to a verifiable OPRF, denoted as the ADDG and vADDG schemes, respectively. The ADDG scheme utilizes TFHE to construct an OPRF candidate: the client outsources homomorphic encryption of a bit string \(\mathbf{x}\), and the server homomorphically computes \(\mathbf z=F(\mathbf{k},\mathbf{x})\) with its PRF key $\bf k$ and returns it to the client. 
    
    We briefly explain how the PRF \(F(\mathbf{k},\mathbf{x})\) is homomorphically calculated in ADDG. For the TFHE ciphertext space \(\mathcal{C}\), we denote by \([x]_P \in \mathcal{C}\) a TFHE encryption of \(x \in \mathbb{Z}_P\). For \(\mathbf{x}=(x_1,\dots,x_k)\in \mathbb{Z}_P^k\), we denote
        \[
            [\mathbf{x}]_P=([x_1]_P,[x_2]_P,\dots,[x_k]_P)\in \mathcal{C}^k.
        \]
    First, the client inputs an encryption of a 128-bit string \(\mathbf{x}\in \mathbb{Z}_2^{128}\). Then, using only homomorphic additions and constant multiplications without bootstrapping, the server homomorphically evaluates $\mathbf{y}\in \mathbb Z_2^{256}$ using the PRF key \(\mathbf{k}\).
    Next, the server performs \({\sf CPPBS}_{(2,3)}\) (Circuit Privacy Programmable Bootstrapping), a programmable bootstrapping with functionalities
        \[
            {\sf CPPBS}_{(2,3)}=\begin{cases}
            [0]_2 \mapsto [0]_3,\\[1mm]
            [1]_2 \mapsto [1]_3,
            \end{cases}
        \]
    to obtain an encryption of \(\mathbf{y}\) over \(\mathbb{Z}_3^{256}\).
    Finally, again using only homomorphic additions and constant multiplications, the server computes an encryption of an 82-trit string \(\mathbf{z}\in \mathbb{Z}_3^{82}\), which is the output of the homomorphic evaluation of the PRF.
    A notable feature of the (v)ADDG scheme is that it consumes only a bootstrapping depth of one: By the removal of the key-switching key, this scheme prevents further bootstrapping executions after the $\sf CPPBS$.
            \vspace{-0.5cm}
            \begin{figure}
                \centering
                    \[
                        \boxed{
                            \eval_{F_{\bf k}}:{[\mathbf{x}]_2^\sk}\in \cc C^{128} {\mapsto} {[\mathbf{y}]_2^\sk} \in \cc C^{256} \overset{{\sf CPPBS}_{(2,3)}}{\mapsto} {[\mathbf{y}]_3^{\sk'}}\in \cc C^{256} \mapsto {[\mathbf{z}]_3^{\sk'}} \in \cc C^{82}
                        }
                    \]
                \caption{Workflow of the Homomorphic PRF suggested in \cite{cryptoeprint:PRF}. The superscripted $\mathsf{sk}$ (resp. $\mathsf{sk'}$) denotes that the secret key is $\mathsf{sk}$ (resp. $\mathsf{sk'}$). $\bf k$ is the PRF key.} 
                \label{fig:enter-label}
                \vspace{-0.5cm}
            \end{figure}

        \subsubsection{Verifiability through Replication.}
        For the extension from ADDG to vADDG for verifiability, \cite{cryptoeprint:PRF} suggested the following method: assume that a client wishes to evaluate \(F(\mathbf{k},\mathbf{x}_i)\) for distinct inputs \(\mathbf{x}_i\) where \(i=1,\dots,\alpha\), and \(\mathbf{k}\) is the server's secret PRF key.

            \begin{enumerate}
                \item First, the server publishes $\kappa$ verification values 
                    \[
                        \mathfrak R=\{(\mathbf{x}^\star_k,\mathbf{z}^\star_k)\}_{1\leq k \leq \kappa} \quad(\mathbf{z}^\star_k=F(\mathbf{k},\mathbf{x}^\star_k))
                    \]
                in plaintext, along with zero-knowledge proofs that enable the client to verify the integrity of these results.
                
                \item The client prepares and encrypts a vector of strings 
                    \[
                        \Big(\overbrace{\mathbf{x}_1,\dots,\mathbf{x}_1}^{\nu \text{ copies} },\dots,\overbrace{\mathbf{x}_\alpha,\dots,\mathbf{x}_\alpha}^{\nu \text{ copies}},\overbrace{\mathbf{x}^\star_{k_1},\dots, \mathbf{x}^\star_{k_\beta}}^{\beta \text{ verification values}}\Big)\in(\bb Z_2^{128})^\gamma
                    \]
                where $(\mathbf{x}^\star_{k_j},\mathbf{z}^\star_{k_j})\overset{\$}{\gets} \mathfrak R$ and $\gamma=\alpha\nu+\beta$.
                
                \item Permute the \(\gamma\) ciphertexts using a random permutation \(\rho: \{1,\dots,\gamma\} \to \{1,\dots,\gamma\}\), and then send them to the server. After the server evaluate the PRF $F(\mathbf{k},\mathbf{x}_i)$ on each input $\mathbf{x}_i$, apply the inverse permutation \(\rho^{-1}\) and decrypt the ciphertexts to obtain \(\gamma\) 82-trit strings 
                \(Z=(\mathbf{z}_s) \in \left(\mathbb{Z}_3^{82}\right)^\gamma \) for \(1\leq s\leq \gamma.\)

                \item Client verifies whether the followings are correct.
                    \begin{enumerate}      
                        \item For sets $\{\mathbf{z}_1,\dots,\mathbf{z}_\nu\}$, $\{\mathbf{z}_{\nu+1},\dots,\mathbf{z}_{2\nu}\},\dots, \{\mathbf{z}_{(\alpha-1)\nu+1},\dots, \mathbf{z}_{\alpha\nu}\}$, each set contains a common single value, and all these $\alpha$ values are distinct.
                        \item $\mathbf{z}_{\alpha\nu+j}=\mathbf{z}^\star_{k_j}$
                    \end{enumerate}
            \end{enumerate}
        The suggested parameters for 80-bit security for verifiability is $(\alpha,\beta,\nu)=(105,10,11)$.

        \subsection{Attack on vADDG}
        The vADDG scheme is designed to prevent forgery attacks by leveraging the assumed hardness of building deep circuits under depth restrictions~\cite{CHLR18}, and by restricting bootstrapping to a single level. In~\cite{cryptoeprint:PRF}, the authors propose a concrete attack strategy under this model and determine secure parameters to resist such attacks, leading to the choice of \((\alpha,\beta,\nu) = (105,10,11)\). However, we show that a successful forgery remains possible under this parameter set with non-negligible probability.
    
        The idea behind this attack is to create a pseudorandom characteristic function and implement it homomorphically. To build such a function, we follow the following procedure.
        Let the adversarial server receives the inputs $\ct_s:=[\mathbf{x}_s]_2$ for all $1\leq s \leq \gamma$ where each $\mathbf{x}_s=(x_{s,u})_{1\leq u\leq 128}\in \bb Z_2^{128}$ is a 128-bit string. 
        Now design a homomorphic characteristic function $\chi_u: \cc C^{128}\to \cc C,  \ct_s\mapsto\overline{\ct_{s,u}}$ as follows:
        First, let an index $u\in \{1,\dots,128\}$ be fixed. Next, for each ciphertext string $\ct_s$, evaluate 
                $$\overline{\ct_{s,u}}:={\sf CPPBS}_{(2,3)}(\ct_{s,u}).$$
        Then, this characteristic function $\chi_u$ is a homomorphic pseudorandom function which outputs $[1]_3$ with probability approximately $1/2$, otherwise $[0]_3$.
        With these $\gamma$ outputs $(\overline{\ct_{s,u}})_{1\leq s\leq \gamma}$, the adversary can try to forge the PRF $F$ by adding $(\overline{\ct_{s,u}})_{1\leq s\leq \gamma}$ to the output strings $\left(\eval_{F(\mathbf{k,-})}(\ct_s)\right)_{1\leq s\leq \gamma}$. For example, for ${\bf t}=(1,0,\dots,0)$, we can forge the PRF by evaluating $$\eval_{F({\bf k},-)}(\ct_s)+{\bf t}\cdot \chi_u(\ct_s).$$
    Note that $\chi_u(\ct_s)$ is the same for the same inputs and so this forgery is not detected if $\chi_u({\ct_{s}})$ is 0 for the all $\beta$ verification values. The success probability is given in Theorem 1.
    
        \begin{algorithm}
            \caption{Attack on vADDG}
            \label{alg: Attack against REP under SHE}
            \begin{algorithmic}[1]
                \Procedure{$\mathsf{vADDG\_Attack}$}{$(\ct_s)_{1\leq s \leq \gamma},u$} \Comment{$u\in \{1,\dots,128\}$ \begin{flushright} $\ct_s=[\mathbf x_s]_2=([x_{s,1}]_2,\dots,[x_{s,128}]_2)$ 
                    for $1\leq s\leq\gamma$.
                \end{flushright}}
                    \State $\mathbf t\gets\bb Z_3^{82}$ for $\mathbf t\neq ({0,0,\dots,0})$ \Comment{Select the trits to forge.}
                    \For{$s = 1$ \textbf{to} $\gamma$}
                    \State $\overline{\ct_{s,u}}\gets {\sf CPPBS}_{(2,3)}([x_{s,u}]_2)$
                    \State $[\mathbf{z}_s]_3\gets\eval_{F(\bf k,-)}(\ct_s)$ \Comment{$[\mathbf z_s]_3=([z_{s,1}]_3,\dots,[z_{s,82}]_3)$.}
                    \State $\ct^{\sf forged}_{s}\gets[{\bf z}_{s}]_3 + \overline{\ct_{s,u}}\cdot \mathbf{t}$  
                    \EndFor
                    \State \textbf{return} $\left(\ct_{s}^{\sf forged}\right)_{1\leq s\leq \gamma}$
                \EndProcedure
            \end{algorithmic}
        \end{algorithm}
        \begin{theorem}
                For a given parameter $(\alpha, \beta, \nu)$, the expected probability that the bootstrapping-depth one circuit in Algorithm~\ref{alg: Attack against REP under SHE} successfully forges the output is 
                    \[
                        \left(1-\left(1/2\right)^\alpha\right)\cdot \left(1/2\right)^\beta.
                    \] over the all random inputs.
        \end{theorem}
        \begin{proof}
                The forgery succeeds when $\overline{\ct_{s,u}}=[0]_3$ for all $s$ correspond to the $\beta$ verification values and at least one $\overline{\ct_{s,u}}=[1]_3$ for $s$ corresponds to $\alpha$ message values. Let $p$ be the probability that $\overline{\ct_{s,u}}$ became an encryption of 1. In Algorithm \ref{alg: Attack against REP under SHE}, the expected value of $p$ over all input values is $1/2$. The attack success probability is obtained by subtracting the probability that all values become zero, $\left(1-p\right)^{\alpha+\beta}$, from the probability that all $s$ corresponding to $\beta$ become zero, $(1-p)^\beta$.  
                Thus, the success probability is given by $Q(p):=\left(1-p\right)^\beta - \left(1-p\right)^{\alpha+\beta}.$
                Substituting $p=1/2$, we have the desired result
                \(
                \left(1-\left(1/2\right)^\alpha\right)\cdot \left(1/2\right)^\beta. 
                \)
        \end{proof} 
        
        \begin{corollary}
            The {\rm vADDG} scheme with parameters $(\alpha, \beta, \nu) = (105, 10, 11)$ provides at most 10-bit security against a malicious adversary.
        \end{corollary}

        In the covert adversary setting, this forgery is detected if at least one of the $\beta$ verification values is altered, which occurs with probability $1 - (1/2)^{\beta}$. Therefore, a covert adversary cannot successfully carry out this attack.

        \subsubsection{Parameter Selection for vADDG.} 
        The success probability of this attack does not change with an increase in the parameter $\nu$. 
        Also, the probability does not change significantly with a decrease in the parameter $\alpha$, since 
            $Q(p) = (1 - (1/2)^\alpha)(1/2)^\beta$
        has a lower bound of 
        $(1/2)\cdot(1/2)^\beta$ and an upper bound of $(1/2)^\beta$, which are both independent of $\alpha$. Thus in our discussion of parameters, we only concerns $\beta$. To ensure $80$-bit security, $\beta$ must exceed 80.
        
        In \cite{cryptoeprint:PRF}, other cheating strategies are mentioned with success probabilities of 
        \(
        1/\binom{\gamma}{\beta} \text{ and } \alpha/\binom{\gamma}{\nu}.
        \)
        By setting \(\beta = 80\), we obtain the following two new parameter sets that satisfy 80-bit security:
        \[
        (\alpha,\beta,\nu) = {(28440, 80, 6)} \text{, or } (189, 80, 10).
        \]
        The first parameter has an overhead size ($(\alpha\nu+\beta )/ \alpha$) of approximately {\(\times 6.0\)} but requires a larger set of inputs, meanwhile the second parameter has an overhead size of approximately {\(\times 10.4\)} but requires a smaller set of inputs.
    

\section{Cryptanalysis on Replication Encoding} \label{sec:attacks on rep}
We recall the Replication Encoding (\REP) and its verification procedure described in \cite{VE}.
For the convenience of presentation, we describe our attack assuming the evaluation of a single-input and single-output polynomial
$f:\mathbb{Z}_t \to \mathbb{Z}_t$ on a ciphertext encrypted under the HE scheme
$\enc: \mathbb{Z}_t^N \rightarrow \mathcal{R}_q^2$, where $t=p$ is a prime.
However, it is straightforward to extend our attack to a multivariate polynomial circuit
$f: \mathbb{Z}_t^\mu \rightarrow \mathbb{Z}_t^\nu$ for $\mu,\nu \in \mathbb{Z}^+$,
where an element of $\mathbb{Z}_t^\mu$ is regarded as an extended slot element.
For the case $t=p^r$ with $r>1$, a similar idea can be applied to construct an analogous attack. We defer the details for $t=p^r$ to Appendix~\ref{Appendix A}.

    \subsection{Replication Encoding}
    Let $f:\mathbb{Z}_t \to \mathbb{Z}_t$ be a circuit that the client requested to evaluate. The \REP{} proceeds as follows. 
        \begin{enumerate}
            \item The client randomly selects a subset $S \subset \{1, \dots, n\}$ with size $|S|=n/2$ for a positive integer $n$ dividing $N$.
            \item The slots indexed by $S^c=\{1,\dots, n\}\setminus S$ are filled with a same message $m$, and the complement slots indexed by $S$ are filled with $n/2$ random values $v_i\overset{\$}{\gets} \bb Z_t$. 
        \end{enumerate}
    We call the slots indexed by $S$ as \emph{verification slots}, and the slots indexed by $S^c$ as \emph{computation slots}. 
    After evaluating $f$ to the ciphertext, its decryption contains $f(v_i)$'s in the verification slots and $f(m)$ in the computation slots.
    For verification, it is assumed that the client knows $f(v_i)$'s and their integrity in advance, either by outsourcing the computation of $f(v_i)$'s without encryption to a third party using existing verifiable computation techniques or by securely computing them himself. Note that the verification values are not published, unlike vADDG. The client verifies as follows:
        \begin{enumerate}
            \item Decrypt the ciphertext to get a vector $(z_i)$ for $i=1,\dots, n$.
            \item Check whether $z_i=f(v_i)$ for verification slots $i\in S$.
            \item Check whether $z_i$ are equal for computation slots $i\in S^c$.
        \end{enumerate}
    If these conditions are satisfied, the common value in the computation slots is accepted as $f(m)$.
    See \cite{VE} for more details.
 
    \subsection{Attacks on \REP} \label{subsec: Attack on the REP}
    Learning or guessing $S\subset\{1,\dots,n\}$ is not easy. Instead of recovering $S$, we take a similar approach to the attack on vADDG: we learn $S$ in \emph{encrypted state} by computing $\enc(J_S)$ for $J_S:=(j_1,\dots,j_n)$ where $j_i=1$ if $i\in S$ and $j_i=0$ otherwise. 
    Once the adversary recovers $\enc(J_S)$\footnote{Note that $J_{S^c} = (1,\dots,1) - J_S$, thus recovering $\enc(J_{S^c})$ is equivalent to recovering $\enc(J_S)$.}, the ciphertext can be easily forged by computing the following:
        \begin{align} \label{ctforged formula}
            \ct_{\textsf{forged}}:=\mathsf{Mult}(\enc(J_S),\eval_f(\ct)) + \mathsf{Mult}(\enc(J_{S^c}),\eval_g(\ct)),
        \end{align}
    where $f$ is the circuit that the client requested to evaluate, and $g$ is any modified circuit by the adversary. 
    Then $\ct_{\textsf{forged}}$ provides the client with a malicious result $g(m)$ but successfully passes the verification step.
    
    To recover $\enc(J_{S^c})$, we designed a circuit to extract the position vector of the common values. 
    Figure \ref{fig: recover S} shows the functionality of the desired algorithm; 
    If the $i$th slot is chosen, the output of this algorithm is the position vector of the slots equal to this value.
    We do not consider the case where $v_i = m$ for some $i \in S$, as in such case the client would already have attained $f(v_i) = f(m)$.

        \begin{figure}
                \centering
                \begin{tikzpicture}
                    \def\squareSize{0.65cm}
                    \node[draw, rectangle, minimum size=\squareSize, fill=gray!50] (a1) at (0, 0) {$m$};
                    \node[draw, rectangle, minimum size=\squareSize] (a2) at (\squareSize, 0) {$v_1$};
                    \node[draw, rectangle, minimum size=\squareSize] (a3) at (2*\squareSize, 0) {$v_2$};
                    \node[draw, rectangle, minimum size=\squareSize] (a4) at (3*\squareSize, 0) {$m$};
                    \node[draw, rectangle, minimum size=\squareSize] (a3) at (4*\squareSize, 0) {$v_3$};
                    \node[draw, rectangle, minimum size=\squareSize] (a4) at (5*\squareSize, 0) {$m$};   
                    \node[draw, rectangle, minimum size=\squareSize] (a3) at (6*\squareSize, 0) {$v_4$};
                    \node[draw, rectangle, minimum size=\squareSize] (aw) at (7*\squareSize, 0) {$m$};
                    
                    \node[draw, rectangle, minimum size=\squareSize] (be) at (10*\squareSize, 0) {$1$};
                    \node[draw, rectangle, minimum size=\squareSize] (b2) at (11*\squareSize, 0) {$0$};
                    \node[draw, rectangle, minimum size=\squareSize] (b3) at (12*\squareSize, 0) {$0$};
                    \node[draw, rectangle, minimum size=\squareSize] (b4) at (13*\squareSize, 0) {$1$};
                    \node[draw, rectangle, minimum size=\squareSize] (b5) at (14*\squareSize, 0) {$0$};
                    \node[draw, rectangle, minimum size=\squareSize] (b6) at (15*\squareSize, 0) {$1$};
                    \node[draw, rectangle, minimum size=\squareSize] (b7) at (16*\squareSize, 0) {$0$};
                    \node[draw, rectangle, minimum size=\squareSize] (b8) at (17*\squareSize, 0) {$1$};
                
                    \draw[->] (aw.east) -- (be.west);

                    \node[draw, rectangle, minimum size=\squareSize] (a1) at (0*\squareSize, -1.1) {$m$};
                    \node[draw, rectangle, minimum size=\squareSize, fill=gray!50] (a2) at (1*\squareSize, -1.1) {${v_1}$};
                    \node[draw, rectangle, minimum size=\squareSize] (a3) at (2*\squareSize, -1.1) {$v_2$};
                    \node[draw, rectangle, minimum size=\squareSize] (a4) at (3*\squareSize, -1.1) {$m$};
                    \node[draw, rectangle, minimum size=\squareSize] (a3) at (4*\squareSize, -1.1) {$v_3$};
                    \node[draw, rectangle, minimum size=\squareSize] (a4) at (5*\squareSize, -1.1) {$m$};    
                    \node[draw, rectangle, minimum size=\squareSize] (a3) at (6*\squareSize, -1.1) {$v_4$};
                    \node[draw, rectangle, minimum size=\squareSize] (cw) at (7*\squareSize, -1.1) {$m$};
                    
                    \node[draw, rectangle, minimum size=\squareSize] (ce) at (10*\squareSize, -1.1) {$0$};
                    \node[draw, rectangle, minimum size=\squareSize] (b2) at (11*\squareSize, -1.1) {$1$};
                    \node[draw, rectangle, minimum size=\squareSize] (b3) at (12*\squareSize, -1.1) {$0$};
                    \node[draw, rectangle, minimum size=\squareSize] (b4) at (13*\squareSize, -1.1) {$0$};
                    \node[draw, rectangle, minimum size=\squareSize] (b5) at (14*\squareSize, -1.1) {$0$};
                    \node[draw, rectangle, minimum size=\squareSize] (b6) at (15*\squareSize, -1.1) {$0$};
                    \node[draw, rectangle, minimum size=\squareSize] (b7) at (16*\squareSize, -1.1) {$0$};
                    \node[draw, rectangle, minimum size=\squareSize] (b8) at (17*\squareSize, -1.1) {$0$};
                
                    \draw[->] (cw.east) -- (ce.west);
                    
                    \node at (8.5*\squareSize, -0.8){$i=2$};
                    \node at (8.5*\squareSize, 0.3){$i=1$};
                \end{tikzpicture}
                \caption{Finding Common Values among $n$ Slots.} \label{fig: recover S}
            \end{figure}

    We first define two functions $\textsf{Duplicate}_i$ and $\textsf{Compare}$ as follows:
        \begin{itemize}
            \item $\textsf{Duplicate}_i: \bb Z_t^n \to \bb Z_t^n$ is a function that maps $(x_1, \dots, x_i, \dots, x_n) \mapsto (x_i, \dots, x_i)$
            \item $\textsf{Compare}(x,y):\bb Z_t^2\to \bb Z_t$ is $1$ if $x=y$ and $0$ if $x\neq y$. We can extend it to act on vectors componentwise manner, $\textsf{Compare}: \bb Z_t^{2n}\to \bb Z_t^n$ 
        \end{itemize}
    Also we extend $\textsf{Compare}$ to act on vectors componentwise manner. 
    We note that Fermat's little theorem guarantees $1-(x-y)^{t-1}=\textsf{Compare}(x,y)$ when $t$ is a prime, which needs $\lceil \log t \rceil$ multiplications to evaluate.
    Now we introduce Algorithm~\ref{alg: common value position} to construct the desired circuit and describe its properties.
    
        \begin{algorithm}
            \caption{Finding Common Value Slots}
            \label{alg: common value position}
            \begin{algorithmic}[1]
            \State \textbf{Input:} A ciphertext $\ct$ encrypted via $\REP$ and an index $i$.
            \State \textbf{Output:} A ciphertext $\ct_{\mathsf{CVS},i}$, possibly an encryption of the vector $J_{S^c}$ if $i \in S^c$.
                \Procedure{$\mathsf{CVS}$}{$\ct,i$} 
                    \State $\ct'\gets\mathsf{\eval_{Duplicate}}_i(\ct)$ \Comment{$\ct'=\enc((x_i,\dots,x_i))$}
                    \State $\ct_{\textsf{CVS},i}\gets\eval_\mathsf{Compare}(\ct',\ct)$
                    \State \textbf{return} $\ct_{\mathsf{CVS},i}$ 
                \EndProcedure
            \end{algorithmic}
        \end{algorithm}
        \begin{theorem}
            For uniformly random $i\overset{\$}{\gets}\{1,\dots,n\}$, {$\mathsf{CVS}(\mathsf{\ct},i)$} in Algorithm \ref{alg: common value position} outputs $\mathsf{Enc}(J_{S^c})$ with probability of $|S|/n=1/2$ in $O(\log t)$ homomorphic multiplications and rotations. 
        \end{theorem}
        \begin{proof}  
            Let $\ct$ be an encryption of a vector that has $m$ in computational slots $S^c$.
            If $i\in S^c$, then $\ct'\gets\eval_{\mathsf {Duplicate}_i}(\ct)$ is an encryption of $(m,\dots,m)$. Consequently, $\mathsf{Eval_{Compare}}(\ct',\ct)$ returns $\enc(J_{S^c})$.
 
            $\eval_{\textsf{Duplicate}_i}$ requires one masking by the $i$-th elementary vector and a partial rotation sum within each message slot of length $n$ that takes $\lceil \log n \rceil$ homomorphic rotation. 
            Also we showed that \eval$_{\textsf{Compare}}$ needs $O(\log t)$ homomorphic multiplications. 
            Since $n\le \phi(t)$, the total complexity is $O(\log t)$ homomorphic multiplications/rotations.
        \end{proof}

        Thus, $\REP$ is not sound in the presence of the malicious adversary with rotation key with index 1.
        \begin{corollary}\label{cor_covert}
            $\REP$ is not sound in the presence of $\cc A_{\sf mal}^{\sf rtk_1}$.
        \end{corollary}

        If the attack fails, (i.e. $i\in S$) then the output of $\ct_{\sf forged}$ of the forgery fomular~\ref{ctforged formula} where $\enc(J_{S^c})\gets\mathsf{CVS}(\mathsf{\ct},i)$ is detected by the client. Thus this attack does not imply that the $\REP$ is not sound in the presence of the covert adversary. 

    \subsubsection{Increasing Attack Success Probability.}
        
    Now we show how to increase the success probability of the attack by repeating Algorithm \ref{alg: common value position} for different indices $i$.
    For any fixed $i$, if we sum up the values in the $n$ slots of $\ct_{\textsf{cvs},i}$, it outputs $|S|$ or $1$ depending on whether $i\in S$ or not. That is, the server can {\em homomorphically distinguish} whether $i\in S$ or not.
    Therefore, an adversary can construct a cheating circuit (Algorithm~\ref{alg: deterministic attack on rep}) that deterministically outputs $\enc(J_{S^c})$. 
    To design such a circuit, first we define three functions $\textsf{RotSum},\textsf{Interpolate}$ and $\textsf{Normalize}_{k}$ as follows: 
        \begin{align*}
            \textsf{RotSum}&:\;(x_1,\dots,x_n) \mapsto \left(\sum_{i=1}^n x_i,\dots, \sum_{i=1}^n x_i\right)\\
            \textsf{Interpolate}&:\begin{cases}
            1\mapsto 0\\ n/2\mapsto 1\\ \text{otherwise anywhere.}\\    
            \end{cases}
            \textsf{Normalize}_k: \begin{cases}
                0\mapsto 0\\ 1,\dots,k \mapsto 1, \\ \text{otherwise anywhere}\\
            \end{cases}
        \end{align*}
    Note that $\eval_{\textsf{RotSum}}$ can be evaluated using $\lceil \log n \rceil$ homomorphic rotations, while $\eval_{\textsf{Interpolate}}$ and $\eval_{\textsf{Normalize}_k}$ can be evaluated using $O(1)$ and $O(\log k)$ homomorphic multiplications, respectively, when the plaintext modulus $t$ is prime.

        \begin{algorithm}
            \caption{Deterministic Recovery of Encryption of $S$}
            \label{alg: deterministic attack on rep}
            \begin{algorithmic}[1]
            \State \textbf{Input:} A ciphertext $\ct$ encrypted via $\REP$ and an index $i$.
            \State \textbf{Output:} A ciphertext $\ct_{\textsf{bool},i}$ which is an encryption of $0^N$ if $i\in S$ and $1^N$ otherwise.
                \Procedure{$\mathsf{is\_in\_set}_S$}{$\ct,i$}
                    \State $\ct_{\textsf{cvs},i}\gets \textsf{CVS}(\ct,i)$
                    \State $\ct_{\textsf{rs}}\gets\textsf{RotSum}(\ct_{\textsf{cvs},i})$ \Comment{$\ct_{\textsf{rs}}=\enc(1,\dots,1)$ or $\enc(n/2,\dots,n/2)$}
                    \State $\ct_{\textsf{bool},i}\gets \eval_{\mathsf{Interpolate}}(\ct_{\textsf{rs}})$
                    \State \textbf{return} $\ct_{\textsf{bool},i}$ \Comment{$\ct_{\textsf{bool},i}=\enc(0,\dots,0)$ or $\enc(1,\dots,1)$}
                \EndProcedure
                    \State
                \State \textbf{Input:} A ciphertext $\ct$ encrypted via $\REP$ and an index $i$, an integer $k$.
                \State \textbf{Output:} A ciphertext $\ct$ possibly encrypting $J_{S^c}$
                \Procedure{$\mathsf{Recover}_{S^c}$}{$\ct,k$} \Comment{Deterministically outputs $\enc(J_{S^c})$ \begin{flushright} when $k=n/2+1$. \end{flushright}}
                    \State $\ct_{\textsf{cvs},i}\gets \textsf{CVS}(\ct,i)$
                    \State $\ct_{\textsf{bool},i}\gets \textsf{is\_in\_set}_S(\ct,i)$
                    \State $\ct'\gets \sum_{i=1}^{k} \sf{Mult}(\ct_{\textsf{bool},i}, \ct_{\textsf{cvs},i})$
                    \State $\ct_{S^c}\gets \eval_{\textsf{Normalize}_{k}}(\ct')$
                    \State \textbf{return} $\ct_{S^c}$
                \EndProcedure
            \end{algorithmic}
        \end{algorithm}

        \begin{theorem}  \label{lem: attack on rep}
            If the rotation key with index 1 is given, $\mathsf{Recover}_{S^c}(\mathsf{ct},k)$ in Algorithm \ref{alg: deterministic attack on rep} can be evaluated within a computational cost of $O(k\log t)$, outputting $\enc(J_{S^c})$ with probability
            \[
            1-\frac{{n/2 \choose k}}{{n \choose k}}.
            \]
            Otherwise, it outputs an encryption of zero vector $\enc(0,\dots, 0)$.
         \end{theorem} 
        \begin{proof}
            Since the slots in $\ct_{\textsf{cvs},i}$ contain $n/2$ ones if $i\in S^c$, and a single one if $i\in S$, applying $\eval_{\textsf{RotSum}}$ yields either 
            $\enc\left(\frac{n}{2},\dots,\frac{n}{2}\right)$ or $\enc(1,\dots,1).$
            Consequently, by interpolation, the output $\ct_{\text{bool},i}$ is a vector of Boolean values, either $(1,\dots,1)$ or $(0,\dots,0)$. Thus, for $i=1,\dots,k$, the sum of the $\ct_{\textsf{cvs},i}$ after homomorphic multiplication with $\ct_{\text{bool},i}$ decrypts to the same value as 
            $|S^c\cap\{1,\dots,k\}|\cdot \enc(J_{S^c}).$
            Since $S\subset\{1,\dots,n\}$ is a random subset of size $n/2$, $|S^c\cap\{1,\dots,k\}|$ is zero with probability $\frac{{n/2 \choose k}}{{n \choose k}}$ and otherwise nonzero.
            Thus, after evaluating $\textsf{Normalize}_{k}$, we obtain $\enc(0,\dots,0)$ with probability $\frac{{n/2 \choose k}}{{n \choose k}}$, or $\enc(J_{S^c})$ with probability $1-\frac{{n/2 \choose k}}{{n \choose k}}$.

            Furthermore, the procedure $\mathsf{Recover}_{S^c}(\ct)$ calls $\textsf{CVS}_i$ for $i=1,\dots,k$, which incurs a cost of $O(k\log t)$ homomorphic operations. The cost of the remaining operations does not exceed $O(k\log t)$.
        \end{proof}

        By taking $k = n/2 + 1$, this attack can achieve a $100\%$ success rate\footnote{With only negligible probability of failure, due to the properties of the underlying RLWE scheme; see Definitions~\ref{def:HE_correctness} and~\ref{def:VHE_Completness}.}.         
    Even when $k < n$, the failure probability of the attack is bounded by $\frac{{n/2 \choose k}}{{n \choose k}} \leq 2^{-k}$, which decreases exponentially with $k$. Therefore, if the adversary is willing to trade off success probability for reduced circuit latency, they may choose a sufficiently small value of $k$. 
    Moreover, since the circuit ${\sf Recover}_{S^c}(\ct,k)$ is parallelizable, increasing $k$ may not lead to increased latency, depending on the adversary's computational capabilities.
    
    Even if the attack fails, the output remains $\enc(0, \dots, 0)$, and thus, by the definition of $\ct_{\sf forged}$ in~\ref{ctforged formula}, the attack remains undetected. Hence, the attack can still be carried out by a covert adversary, even when $k$ is sufficiently small.

        \begin{corollary}\label{Cor: Rep}
            The scheme $\mathsf{REP}$ is not secure against the covert adversary $\mathcal{A}_{\mathsf{cov}}^{\rtk_1}$.
        \end{corollary}

        The following result follows directly from Corollaries~\ref{cor_covert} and~\ref{Cor: Rep}.

        \begin{corollary}
            If BFV bootstrapping requires the key $\rtk_1$, then $\mathsf{REP}$ is either not bootstrappable or insecure against covert/malicious adversaries.
        \end{corollary}

    \subsection{Attack on Replication Encoding with Multiple Secret Keys}\label{subsec: diff keys}

        \subsubsection{Patch on \REP{} using Multiple Secret Keys.}
        We introduce $\REP^{\sf msk}$, a countermeasure to the previous deterministic attack. The previous attack evaluates homomorphic rotations to compare values in different slots and recover $\enc(J_{S^c})$. 
        Notably, a rotation key with index 1 is essential for executing this attack. Meanwhile, in $\REP$, a single value $m$ is encoded into $n$ slots. Therefore, to perform operations on two values $m_1$ and $m_2$ encoded via $\REP$, only rotation keys whose index is a multiple of $n$ will be required. Nevertheless, we note that a rotation key with index 1 remains essential for the bootstrapping of homomorphic encryption. Thus, the idea of preventing the above attack by not providing a rotation key with index 1 is not adequate. 
        
        Instead of restricting the rotation key with index 1, we consider the following possibility: Replicate the messages into $n$ different ciphertexts encrypted with \emph{different keys}. This would prevent interoperability between ciphertexts, making it difficult for the adversary to homomorphically determine in which slots the messages belong to.
        First, prepare a random subset $S \subset \{1, \dots, n\}$ with $|S| = n/2$, a message $m \in \bb Z_t$, and verification values $v_i \gets \bb Z_t$ for $i \in S$. 
        Next, generate $n$ different keys $\{\pk_i, \sk_i, \mathsf{evk}_i\}$ for $i \in \{1, \dots, n\}$.
        Now, instead of encrypting messages as a vector, we encrypt them elementwisely, i.e. $\ct_i = \enc_{\pk_i}(m)$ for $i \in S$ and $\ct_i = \enc_{\pk_i}(v_i)$ for $i \notin S$.
        Then, the server can compute those ciphertexts with different $\evk_i$'s.
        This patch achieves the same effect as the original \REP{} while preventing the previous attack in Algorithm~\ref{alg: deterministic attack on rep} by restricting the homomorphic comparison, while still providing the bootstrapping functionality. 
        
            \begin{figure}
                \centering
                \begin{tikzpicture}
                    \def\squareSize{0.7cm}
                    \def\adjust{0cm}
                    \node[draw, rectangle, minimum size=\squareSize] (a1) at (10.2*\squareSize-\adjust, 2) {$m$};
                    \node at (10.2*\squareSize-\adjust, 1.35) {$\pk_1$};
                    \node[draw, rectangle, minimum size=\squareSize] (a2) at (11.4*\squareSize-\adjust, 2) {$m$};
                    \node at (11.4*\squareSize-\adjust, 1.35) {$\pk_2$};
                    \node[draw, rectangle, minimum size=\squareSize] (a3) at (12.6*\squareSize-\adjust, 2) {$v_3$};
                    \node at (12.6*\squareSize-\adjust, 1.35) {$\pk_3$};
                    \node[draw, rectangle, minimum size=\squareSize] (a4) at (13.8*\squareSize-\adjust, 2) {$v_4$}; 
                    \node at (13.8*\squareSize-\adjust, 1.35) {$\pk_4$};
                    
                    \node[draw, rectangle, minimum size=\squareSize] (a3) at (3*\squareSize, 2) {$m$};
                    \node[draw, rectangle, minimum size=\squareSize] (a4) at (4*\squareSize, 2) {$m$}; 
                    \node[draw, rectangle, minimum size=\squareSize] (a3) at (5*\squareSize, 2) {$v_3$};   
                    \node[draw, rectangle, minimum size=\squareSize] (aw) at (6*\squareSize, 2) {$v_4$};  
        
                    \node at (4.6*\squareSize-\adjust, 1.35) {$\pk$};
                    
                    \draw[-stealth] (4.8, 2) -- (6.5, 2);
                \end{tikzpicture}
                \vspace{-0.5cm}
            \end{figure}
        
        However, we can also attack $\REP^{\sf msk}$ similarly to our attack on vADDG, namely, by implementing a pseudorandom characteristic function. It is possible since the evaluation does not rely on homomorphic rotation or comparison but only on slot-wise operations. 
        We analyze how this attack applies to this variant. Moreover, in vADDG, it was able to patch the attack by increasing the number of verification values by adjusting some parameters. We intend to investigate whether the same approach can be applied to patch $\REP^{\sf msk}$ or not.

        \subsubsection{Attack on $\REP^{\sf msk}$.} \label{subsec:3.3 Extended Attack}
        First, let the adversary fix a random subset $A\subset \bb Z_t$ where the message $m$ is expected to belong. 
        Next, evaluate $\chi_A$ homomorphically on a ciphertext.
        After evaluating $\chi_A$, it becomes possible to forge the values that exactly belong to $A$.
        If $\{m,v_i:i\in S\}\cap A = \{m\}$, then ${\sf CMult}$ between $\ct_S:=\eval_{\chi_A}(\ct)$ and the encoding of the vector $(\overbrace{1\dots1}^{n}\Vert\overbrace{0\dots0}^{N-n})$ is an encryption of a vector $(J_{S^c}\,\Vert\,\overbrace{0\dots0}^{N-n})$.
        
            \begin{algorithm}
                    \caption{Extended Attack on \REP}
                    \label{alg: probabilistic attack on rep}
                    \begin{algorithmic}[1]
                    \State \textbf{Input:} A ciphertext $\ct$ encrypted via $\REP$ and a subset $A\subset \bb Z_t$.
                    \State \textbf{Output:} A ciphertext possibly encrypting $J_{S^c}$ in the first $n$ slots.           
                        \Procedure{$\mathsf{Extended\_Attack}$}{$\ct, A$}
                            \State $\ct_{\textsf{bool}}\gets \textsf{Eval}_{\chi_A}(\ct)$
                            \State $\mathsf{pt}_{\textsf{Masking}}\gets (\overbrace{1,\dots,1}^{n},\overbrace{0,\dots,0}^{N-n})$
                            \State $\ct_{S^c}\gets\textsf{CMult}(\ct_{\sf bool},\mathsf{pt}_{\sf Masking})$ 
                            \State \textbf{return} $\ct_{S^c}$
                        \EndProcedure
                    \end{algorithmic}
            \end{algorithm}

        Let us calculate the attack success probability with respect to the size of $A$. 
        Since $A$ is chosen randomly, the success probability of the attack depends only on the size of $A$. 
        Let $p=|A|/{t}$ be the probability of each element $m$ or $v_i$ belonging to the set $A$. 
        Set $q := 1 - p$ and $\mu:=\left|\{m,v_i:i\in S\}\right|=1+\frac{n}{2}$. 

            \begin{theorem}\label{thm2: univ}
                Let $A$ be a randomly chosen subset of $\bb Z_t$ with size $|A|=p\cdot t$. For $\ct_{S^c}\gets\mathsf{Extended\_Attack}(\ct,A)$ in Algorithm \ref{alg: probabilistic attack on rep}, the probability $$\Pr\left[\mathsf{Dec}(\ct_{S^c})=(J_{S^c}\Vert\overbrace{0\dots0}^{N-n})\right]=:Q(p)$$ is maximized when $p=\mu^{-1}=\frac{1}{1+(n/2)}$. Also $Q(p)>e^{-1}p$ when restricted to $p=\mu^{-1}.$
            \end{theorem}
            \begin{proof}
                Let's define $\cc V:=\{m,v_i:i\in S\}$. Then, there are three possible cases: 
                    \begin{itemize} 
                        \item Case (1): $\cc V\cap A=\emptyset$. The probability is $q^{\mu}$.
                        \item Case (2): $\{v_i:i\in S\}\cap A\neq\emptyset$. The probability is $1-q^{\mu-1}$. 
                        \item Case (3): $\cc V \cap A=\{m\}$. The probability is $q^{\mu-1}-q^{\mu}$. 
                    \end{itemize}
                Nothing happens in the case (1), and the adversary is caught in the case (2). The case (3) is the case of a successful attack. 
                Now represent the probability of case (3) $Q(p)$ in terms of $p$: $Q(p)=q^{\mu-1}-q^\mu  =(1-p)^{\mu-1}-(1-p)^{\mu}.$
                
            By simple calculus, one can verify that $\frac{dQ}{dp}(\mu^{-1})=0$ and $Q(p)$ is maximized when $p=\mu^{-1}$: $Q(\mu^{-1})=(1-\mu^{-1})^{\mu-1}-(1-\mu^{-1})^{\mu}=(1-p)^{(1/p)-1}-(1-p)^{(1/p)}.$
            
            Now let $R(p):=(1-p)^{(1/p)-1}-(1-p)^{(1/p)}$. Then, $\lim_{p\to 0+}\frac{dR}{dp}(p)=e^{-1}$ and $\frac{d^2R}{dp^2}>0$. Thus 
            $R(p)>e^{-1}p$.
            \end{proof}
     
        Therefore, for any given parameter $\mu=1+(n/2)$, if the adversary randomly chooses a subset $A$ so that $p=|A|/t$ is approximately $\mu^{-1}$, and homomorphically evaluates $\chi_A$, then the adversary's advantage is maximized with a probability at least $({e \mu})^{-1}$.
        
        Thus, unlike in the vADDG attack where a countermeasure was possible by adjusting parameters, in this variant of \REP{} a forgery attack can be mounted with non-negligible probability. However, under the assumption of a covert adversary model, such an attack is not feasible. 

        \paragraph{Pseudorandom Characteristic Function over $\bb Z_t$.}
        The attack cost in Theorem~\ref{thm2: univ} is determined by the evaluation cost of the characteristic function. In the attack on vADDG, the evaluation was straightforward since extracting encrypted bits from the TFHE ciphertext string was easy. However, homomorphically evaluating \(\chi_A\) for any randomly given subset \(A\) with the full support \(\mathbb{Z}_t\) generally requires \(O(\sqrt{t})\) homomorphic operations, which results in an exponential cost with respect to \(\log t\). Therefore, it is necessary to devise an efficient (polynomial time) method for evaluating a pseudorandom characteristic function over \(\mathbb{Z}_t\).
       
        Before constructing a pseudorandom characteristic function, note that the pseudorandom characteristic function used in this attack does not need to be cryptographically secure; rather, it is sufficient for it to be a function that heuristically produces an unbiased uniform distribution.
        
        To construct an efficient pseudorandom characteristic function, we assume the pseudorandomness of the distribution of the roots of unity in $\bb Z_t$:
        For an integer $t$ and $d \mid \phi(t)$, let $U_{t,d}:=\{x\in \bb Z_t^\times:x^{\phi(t)/d}=1\} \subset \bb Z_t^\times$. We assume that this subset $U_{t,d}$ of size $\phi(t)/d$ is an unbiased sample from a uniform distribution of $(\phi(t)/d)$-combinations from $\bb Z_t^\times$. Specifically, for a large odd prime $p$, $U_{p,2}$ is the set of quadratic residues modulo $p$, whose pseudorandomness is utilized to construct the Legendre PRF~\cite{LPRF}. 
        Although the specific assumptions and parameters in here are different from those of Legendre PRF, we again note that it is sufficient for it to be a function that heuristically produces an unbiased uniform sampling.
        
        Now we state the property of the characteristic function $\chi_{U_{t,d}}: \bb Z_t\to\{0,1\}$ for $t=p$.
        For $t=p^r, r>1$, refer Appendix~\ref{Appendix A}. 
        
        \begin{lemma}\label{bfvprf} For $t=p$ and $d \mid p-1=\phi(t)$, we have
            $$d\cdot \chi_{U_{t,d}}(x)=\sum_{i=0}^{d-1} \left(x^{\phi(t)/d}\right)^i.$$
        \end{lemma}
        \begin{proof}
        \begin{enumerate}
            \item First, suppose that $x=0$. Then the left-hand side $(\sf LHS)$ and the right-hand side $(\sf RHS)$ is $0$.
            \item Next, suppose that $x\in U_{t,d}$. Since $x^{\phi(t)/d}=1$ by definition of $U_{t,d}$, both $\sf LHS$ and $\sf RHS$ are equal to $d$.
            \item Lastly, suppose that $x\in \bb Z_t^\times\setminus U_{t,d}$. Then, $x^{\phi(t)/d}\neq 1$ by definition of $U_{t,d}$ Also, $(x^{\phi(t)/d})^d-1=(x^{\phi(t)/d}-1)\left(\sum_{i=0}^{d-1} \left(x^{\phi(t)/d}\right)^i\right)=0$.
            Thus, $\sum_{i=0}^{d-1} \left(x^{\phi(t)/d}\right)^i=0$ since $\bb Z_t$ is an integral domain. 
        \end{enumerate} 
        \end{proof}
        Note that this function can be homomorphically evaluated using $O(\log t)$ homomorphic multiplications with a multiplication depth of $\lceil\log\phi(t)\rceil$.
        
         By homomorphically evaluating $\chi_{U_{t,d}}(x+a)$ for a random integer $a$ and an appropriate divisor $d \mid (p-1)$ with $d \approx \mu$, the adversary can construct a characteristic function $\chi_A$ whose support has size $|A| = \phi(t)/d \approx \lfloor t / \mu \rceil$, thereby achieving the maximal attack probability in Theorem~\ref{thm2: univ}.
        If necessary, the adversary may evaluate several characteristic functions to construct another size of characteristic function.\footnote{For example, one may utilize the formula $\chi_A(x) \cdot \chi_B(x) = \chi_{A \cap B}(x).$}

        \begin{corollary}
            $\sf REP^{\sf msk}$ is not secure against the presence of $\cc A_{\sf mal}.$
        \end{corollary} 
        Whether there exists a method to successfully execute a forgery attack with overwhelming probability remains an open question.
         

\section{Cryptanalysis on Polynomial Encoding} \label{sec:PE}

    \subsection{Polynomial Encoding}
    There is another suggested scheme for verifiable HE in \cite{VE}, the Polynomial Encoding (\PE).
    Its plaintext and ciphertext spaces are inherited from the original homomorphic encryption scheme, augmented with a new indeterminate $Y$.
    For example, in the BFV scheme, its encryption is 
        \[
            \bb{Z}^N_t[Y] \overset{\enc_\PE}{\longrightarrow} \mathcal{R}^2_q[Y],
        \]
    where the $\enc_\PE$ is the coefficient-wise encryption.
    To avoid confusion, we represent the elements of $\mathcal{R}^2_q[Y]$ using capitalized sans-serif font: $\mathsf{F}(Y) = \ct_0 + \ct_1 Y + \dots + \ct_d Y^d$. 
    We refer to $\mathsf{F}(Y)$ as a ciphertext polynomial and each of its coefficients as a ciphertext.

    Encryption in \PE{} proceeds as follows: 
    For the number of slots $N$, a random verification value $\mathbf{v}=(v_1,\dots, v_N)\in \bb Z_t^N$ and a secret value $\alpha \in \mathbb{Z}_t^\times$ are chosen. 
    Then, the message $\mathbf{m}=(m_1,\dots,m_N)\in \bb Z_t^N$ and the verification value $\mathbf{v}$ are interpolated at $Y = 0$ and $Y = \alpha$, respectively. 
    After that, each message is encrypted coefficientwise:
    
        \begin{center}
            \begin{tikzpicture}
                  \matrix (m) [matrix of math nodes,row sep=1em,column sep=3em,minimum width=2em]
                        {
                            \bb Z^N_t & \bb Z^N_t[Y] & \cc R^2_q[Y]\\
                            \mathbf{m}& \mathbf{m}+\left[\frac{\mathbf{v-m}}{\alpha}\right]_tY&\ct_0+\ct_1Y\\
                        };
                  \path[]
                        (m-1-1) edge [-stealth] (m-1-2)
                        (m-1-2) edge [-stealth] (m-1-3)
                        (m-1-1) edge [-stealth, out=15, in=165] node[above]{$\enc_\PE$} (m-1-3)
                        
                        (m-2-1) edge [|-stealth, out=-15, in=195] node[below] {$\enc_\PE$}(m-2-3)
                        (m-2-2) edge [|-stealth] (m-2-3)
                        (m-2-1) edge [|-stealth] (m-2-2);
            \end{tikzpicture}
        \end{center}
        
    The server coefficientwise performs homomorphic operations like addition, constant multiplication, rotation, and even bootstrapping, except for multiplication. 
    The multiplication is evaluated as a polynomial in $Y$:
    For example, the multiplication between two ciphertext polynomials $\ct_0+\ct_1Y$ and $\ct_2+\ct_3Y$ is 
        \[
            \mathsf{Mult}(\ct_0,\ct_2)+\mathsf{Mult}(\ct_1,\ct_2)Y+\mathsf{Mult}(\ct_0,\ct_3)Y+\mathsf{Mult}(\ct_1,\ct_3)Y^2.
        \]
    If the server has performed the operations correctly, the result $\mathsf F(Y)\in \cc R^2_q[Y]$ satisfies $\dec(\mathsf F(0)) = f(\mathbf{m})$ and $\dec(\mathsf F(\alpha)) = f(\mathbf{v})$. 
    In the verification step, the client checks whether $\dec(\mathsf F(\alpha))=f(\mathbf{v})$ and if it is correct, accepts $\dec(\mathsf F(0))$ as $f(\mathbf{m})$.
    
        \subsubsection{Re-Quadratization Protocol.} 
        As the computation progresses, the degree of $Y$ increases exponentially, leading to significant computational overhead. To mitigate this, \cite{VE} proposed a client-assisted computing protocol called the Re-Quadratization (\REQ) protocol.
        This protocol makes a quartic polynomial $\mathsf{Q}_4(Y)\in \cc R^2_q[Y]$ into a quadratic polynomial $\mathsf{Q}_2(Y)\in \cc R^2_q[Y]$ while ensuring that $\dec(\mathsf{Q}_4(0))=\dec(\mathsf{Q}_2(0))$.
    
        However, if $\dec(\mathsf{Q}_4(\alpha)) = \mathbf{a}$, 
        then $\dec(\mathsf{Q}_2(\alpha)) = \mathbf{a} + \mathbf{r}$ 
        for some uniform random blinding vector $\mathbf{r} \in \mathbb{Z}_t^N$.
        Because of this random blinding vector $\mathbf{r}$, the client must compute and handle the deviations of the circuit introduced by $\mathbf{r}$.
        For instance, when performing squaring 
        \(
        (\mathbf{a} + \mathbf{r})^2 
        = \mathbf{a}^2 + 2\mathbf{a}\mathbf{r} + \mathbf{r}^2,
        \)
        the user must compute the deviation $2\mathbf{a}\mathbf{r} + \mathbf{r}^2$ and subtract it to recover $\mathbf{a}^2$ in the \REQ{} protocol. This self-correctness property through client-side computation forces the client to perform as much computation as if the client itself were running the delegated computation in plaintext, negating much of the intended benefit of HE.
        
        Despite this drawback, the \REQ{} protocol reduces the computational overhead of homomorphic computation via \PE.
        If the server multiplies two quadratic polynomials $\mathsf{Q_2}(Y)$, $\mathsf{Q}'_2(Y)$ and gets a quartic polynomial $\mathsf{Q_4}(Y)$, then the server can use \REQ{} to reduce $\mathsf{Q}_4(Y)$ back to a quadratic polynomial $\mathsf{Q}''_2(Y)$. Therefore the server can evaluate a deep circuit avoiding an exponential overhead.
        See V.D. and Appendix E in \cite{VE} for details.
        Here, we only use \REQ{} protocol as an oracle for avoiding overhead in deep circuit evaluation.
      
    \subsection{Attack on Polynomial Encoding}
    Now we present the attack on the Polynomial Encoding. We assume that the \REQ{} protocol serves as a subroutine in the computation. 
    
    Let $f(x_1,\dots, x_k)$ be an honest polynomial circuit that operates on the ciphertext polynomial space $\mathcal{R}^2_q[Y]$.
    Given $k$ inputs $\mathsf{L}_i(Y) = \ct^{(0)}_i + \ct^{(1)}_i Y$ for $i=1,\dots, k$, the polynomial circuit $f$ processes these inputs and produces the output
    \(
        f(\mathsf{L}_1(Y), \dots, \mathsf{L}_k(Y)) = \mathsf{F}(Y) \in \mathcal{R}^2_q[Y].
    \)
    When evaluated at $Y = \alpha$, we obtain
    \[
        \mathsf{F}(\alpha) = \eval_f(\ct_1^{(0)}+\alpha\cdot\ct_1^{(1)}, \dots, \ct_k^{(0)}+\alpha\cdot\ct_k^{(1)}) = \enc(f(v_1, \dots, v_k)).
    \]
    Now let $g$ be any malicious polynomial circuit that outputs $\mathsf{G}(Y)$. Similarly, when evaluated at $Y=0$, it has
    \[
        \mathsf{G}(0) = \eval_g(\ct_0^{(0)}, \dots, \ct_k^{(0)}) = \enc(g(m_1, \dots, m_k)).
    \]
    Since only the result $\sf G(0)$ will be needed, the circuit $g$ does not need to act on the entire ciphertext polynomial; it only needs to operate on its constant term, namely $\ct_0^{(0)}, \dots, \ct_k^{(0)}$. Consequently, the evaluation of $g$ does not require \REQ{}.
    To forge the result, we introduce a trick: define $\mathsf{L}^\star(Y):= \enc(u_1,\dots,u_N) Y$ where all $u_i\in \bb Z_t^\times$. For the polynomial circuit $p(y) := y^{\phi(t)}$, evaluating $p$ at $\mathsf{L}^\star(Y)$ yields $\mathsf P(Y)$ which satisfies
    \[
        \mathsf{P}(0) = p(\mathsf{L}^\star(0)) = \enc(0,\dots,0), \quad \mathsf{P}(\alpha) = p(\mathsf{L}^\star(\alpha)) = \enc(1,\dots,1). 
    \]
    Thus, if we evaluate the following polynomial circuit
    \[
        h(y_1, \dots, y_k, \ell^\star) := g(\tilde y_1, \dots, \tilde y_k) + (f(y_1, \dots, y_k) - g(\tilde y_1, \dots, \tilde y_k)) p(\ell^\star)
    \]
    at $\left(y_1,\dots,y_k,\ell^\star\right)\gets\left(\mathsf{L_1}(Y),\dots,\mathsf L_{k}(Y),\mathsf{L}^\star(Y)\right)$ where $\tilde y_i=\mathsf{L}_i(0)$ for $y_i=\mathsf{L}_i(Y)$,
    the output $\mathsf{H}(Y)$ satisfies
    \[
        \dec(\mathsf{H}(0)) = \dec(\mathsf{G}(0)), \quad \dec(\mathsf{H}(\alpha)) = \dec(\mathsf{F}(\alpha))
    \]
    which successfully passes the verification and gives the forged results.

        \begin{algorithm}
            \caption{Attack against \PE}
            \label{alg: Attack against PE}
            \begin{algorithmic}[1]
            \State \textbf{Input:} $\{\mathsf L_i(Y)\}_{1\leq i\leq k}$ are input points and  $\mathsf F(Y)$ is a honest result.
            \State \textbf{Output:} ${\sf H}(Y)$, a forged ciphertext polynomial.
                \Procedure{$\mathsf{PE\_Attack}$}{$\{\mathsf L_i(0)\}, \mathsf F(Y)$} 
                    \State $\ct_{\sf forged}\gets \eval_g(\mathsf L_1(0),\dots, \mathsf L_k(0))$ \Comment{$g$ is any circuit different from $f$.}
                    \State $\mathsf L^\star(Y)\gets\enc(u_1,\dots,u_N)Y$ \Comment{Every $u_i$ is a unit.}
                    \State $\mathsf P(Y)\gets p({\sf L}^\star (Y))$ \Comment{$p(y)=y^{\phi(t)}$.}
                    \State $\mathsf H(Y)\gets \ct_{\sf forged}+(\mathsf{F}(Y)-\ct_{\sf forged})\cdot \mathsf{P}(Y)$
                    \State \textbf{return} $\mathsf H(Y)$
                \EndProcedure
            \end{algorithmic}
        \end{algorithm}
    
        \begin{theorem}
            If the adversary can access the $\REQ$ protocol, Algorithm \ref{alg: Attack against PE} can be executed within the cost of $O(\log t)$ with deterministic output ${\sf H}(Y)$.
        \end{theorem}
        \begin{proof}
            The correctness of the attack circuit is given above.
            Take $g$ to be a constant circuit different from $f$. Then the compuational cost is determined by evaluating $p(y)=y^{\phi(t)}$ at $y=\mathsf L^\star (Y)$, which requires $\lceil\log(\phi(t))\rceil$ multiplications in $\cc R^2_q[Y]$ and $\REQ$ protocol. Therefore the total computational complexity is $O(\log t)$.
        \end{proof}

        \begin{corollary}
            $\PE$ is not secure against the presence of $\cc A_{\sf cov}^{\cc O_{\sf ReQ}}.$
        \end{corollary}

        However, this attack can be detected by counting the number of $\sf ReQ$ calls.
        The honest output is ${\sf F}(Y)$, and in order to compute $H(Y)$, one must additionally compute ${\sf P}(Y)$ and then multiply it by ${\sf F}(Y) - \ct_{\sf forged}$.
        Therefore, compared to the honest circuit, the number of $\sf ReQ$ calls increases by at least $\lceil \log \phi(t) \rceil$.
        Thus, in this case, the adversary needs to deceive the client about the number of $\REQ$ calls required to evaluate the honest circuit.
        
%
%

\section{Attack Implementation on $\REP$ and $\REP^{\sf msk}$}\label{sec: Implementation}

We implemented our attacks in the \textsc{Veritas} library, which is based on the Lattigo BFV Implementation~\cite{Lattigo}. Specifically, we implemented the attack on $\REP$ (Algorithm~\ref{alg: deterministic attack on rep}) as well as the attack on $\REP^{\sf msk}$ (Algorithm~\ref{alg: probabilistic attack on rep}). Since the Lattigo library currently does not support BFV bootstrapping, we modified the value of $t$ in the parameters provided by the example code in the \textsc{Veritas} library to ensure that the forgery circuit operates without bootstrapping. However, when bootstrapping is supported, the attack remains feasible regardless of the value of $t$.
\footnote{For \PE, \(t\) must be of size roughly equal to the security parameter \(\lambda\), since the secret \(\alpha\) is chosen from \(\mathbb{Z}_t^\times\). With this parameter setting, the implementation of the attack is not feasible without bootstrapping. Thus, we did not implement attack on \PE~due to the lack of bootstrapping functionality.}. 
Specific parameters are given in Table~\ref{tab: Params}. 

As for vADDG, since the verifiable variant is not implemented, we did not implement the attack. However, since the attack circuit is very shallow and simple, and in particular very similar to the attack on $\REP^{\sf msk}$, the feasibility of the attack on vADDG is straightforward even without a demonstration via implementation. 

\subsection{Attack Description}
The attack on $\REP$ and $\REP^{\sf msk}$ aims to implement the following circuit:
\[
\mathsf{Mult}(\eval_f(\ct), \enc(J_S)) + c \cdot \enc(J_{S^c})
\]
for some nonzero constant $c$.
Since $\enc(J_S) = (1,\dots,1) - \enc(J_{S^c})$, the attack reduces to recovering a scalar multiple of $\enc(J_{S^c})$.
Therefore, the attack cost corresponds to the amount of resources required to recover a scalar multiple of $\enc(J_{S^c})$.

\subsection{Attack Results}
The experiments were performed on an Intel(R) Xeon(R) Silver 4114 CPU at 2.20GHz running Linux in a single-threaded environment. The attack results are summarized in Tables~\ref{tab:Cost},~\ref{tab:Prob}. The time cost only measured evaluation time for $\enc(J_{S^c})$. 

\subsubsection{Attack on \REP.}
In experiment 1 to 4, we implement the attack on $\REP$ described in Algorithm \ref{alg: deterministic attack on rep} with various $k=1,2,4,33$. We generated keys whose rotation indices are divisors of \( n \), assuming a scenario in the adversarial server needs to bootstrap with these rotation keys.
We optimized the attack circuit described in Algorithm~\ref{alg: deterministic attack on rep} as follows. 
Instead of evaluating \(\eval_{\sf Interpolate}\), which maps \(\enc(1,\dots,1)\) to \(\enc(0,\dots,0)\) and \(\enc\left(n/2,\dots, n/2\right)\) to \(\enc(1,\dots,1)\), 
we directly performed the subtraction \(\ct \mapsto \ct - \enc(1, \dots, 1)\). 
This operation maps \(\enc(1,\dots,1)\) to \(\enc(0,\dots,0)\) and \(\enc\left(n/2,\dots,n/2\right)\) to \(\enc\left(n/2 - 1,\dots,n/2-1\right)\). 
As a result, we obtain a ciphertext that corresponds to a constant multiple of \(J_{S^c}\).

According to Theorem~\ref{lem: attack on rep}, the expected attack success probabilities given by \( 1 - \binom{32}{k} / \binom{64}{k} \) are \(50\%\), \(75.40\%\), \(94.34\%\), and \(100\%\) for \(k = 1\), \(2\), \(4\), and \(33\), respectively.
This attack circuit can be implemented without bootstrapping under the parameters 
\((\log N = 15, \log PQ = 880, \log Q = 759, t = 2^{16} +  1)\). 
As shown in Table~\ref{tab:Cost}, the attack time increases linearly with $k$, but
we note that Algorithm~\ref{alg: deterministic attack on rep} is inherently parallelizable,  
as it performs ${\sf CVS}(\ct, i)$ and ${\sf is\_in\_set}(\ct,i)$ independently for each \(i = 1, \dots, k\).
Thus, by parallelizing the execution of ${\sf CVS}(\ct, i)$, the overall attack time can be reduced to under one minute even with $k=33$.
In addition, leveraging GPU-based optimizations could potentially reduce the runtime to just a few seconds. 
This assumption is realistic, given that adversaries in the context of HE are typically assumed to have access to significant computational resources.

\subsubsection{Attack on $\REP^{\sf msk}.$}
We implement the attack on $\REP^{\sf msk}$ described in Algorithm \ref{alg: probabilistic attack on rep}. Since this attack naturally applies to $\REP$, without any modification to the $\REP$ scheme we implemented this attack: Specifically, we did not generate the $\rtk_1$.

Since \(\phi(t) = 2^{16}\), when constructing the characteristic function in Lemma~\ref{bfvprf}, 
we can choose \(d\) to be a power of 2. 
In Experiment~5, where \(n = 64\), we set \(d = 32 \cong \mu = 33\) for the characteristic function. 
As a result, the output becomes a scalar multiple of \(\enc(J_{S^c})\) with non-negligible probability.
According to Theorem~\ref{thm2: univ}, the corresponding expected attack success probability is $\left(1-\frac{2^{11}}{2^{16}+1}\right)^{32} - \left(1-\frac{2^{11}}{2^{16}+1}\right)^{33} \cong 1.13\%$. Also in Experiment 6, where \(n=32\), we set \(d=16 \cong \mu=17\). The corresponding expected attack success probability is $\left(1-\frac{2^{12}}{2^{16}+1}\right)^{16} - \left(1-\frac{2^{12}}{2^{16}+1}\right)^{17} \cong 2.22\%$. The parameters and results are presented in Table~\ref{tab: Params},\ref{tab:Cost} and \ref{tab:Prob}.

\begin{table}[htb!]
\centering
\begin{tabular}{|cc|c|c|c|c|c|c|}
\hline
\multicolumn{2}{|c|}{Parameters} & $\log N$ & $\log PQ$ & $\log Q$ & \makecell{$t$} & \makecell {$|S|=|\{1,\dots,n\}|$} & $k$\\ \hline \hline
\multicolumn{1}{|c|}{\multirow{4}{*}{\REP{} Attack}} & Exp1 &  \multirow{4}{*}{15} & \multirow{4}{*}{880} & \multirow{4}{*}{759} & \multirow{4}{*}{$2^{16}+1$} &\multirow{4}{*}{64} & 1 \\ \cline{2-2} \cline{8-8}
\multicolumn{1}{|c|}{} & Exp2 &  &  &  & & & 2 \\ \cline{2-2} \cline{8-8} 
\multicolumn{1}{|c|}{} & Exp3 &  &  &  & & & 4 \\ \cline{2-2} \cline{8-8}
\multicolumn{1}{|c|}{} & Exp4 &  &  &  & & & 33 \\ \cline{1-2} \cline{7-8} \hline\hline

\multicolumn{1}{|c|}{\multirow{2}{*}{\makecell{$\REP^{\sf msk}$Attack}}} & Exp5 &  \multirow{2}{*}{15} & \multirow{2}{*}{880} & \multirow{2}{*}{759} & \multirow{2}{*}{$2^{16}+1$} & 64 & -\\ \cline{2-2} \cline{7-8} 
\multicolumn{1}{|c|}{} & Exp6 &  &  &  &  & 32 & -\\ \hline
\end{tabular}
\caption{Parameters Selection}
\label{tab: Params}
\begin{tabular}{|cc|c|c|}
\hline
\multicolumn{2}{|c|}{\makecell{Evaluation Time\\(Sec/op)}} & \makecell{Attack Time \\(Evaluating $\enc(J_{S^c})$)} & Remark\\ \hline\hline

\multicolumn{1}{|c|}{\multirow{4}{*}{\makecell{$\REP$ Attack}}} & Exp1  & 8.37s & \multirow{4}{*}{\makecell{$\rtk_1$ generated.}}\\ \cline{2-3}
\multicolumn{1}{|c|}{} & Exp2  & 15.91s & \\ \cline{2-3}
\multicolumn{1}{|c|}{} & Exp3  & 29.67s & \\ \cline{2-3}
\multicolumn{1}{|c|}{} & Exp4  & 246.54s & \\ \hline\hline

\multicolumn{1}{|c|}{\multirow{2}{*}{\makecell{$\REP^{\sf msk}$Attack}}} & Exp5  & 6.05s  & -\\ \cline{2-4} 
\multicolumn{1}{|c|}{} & Exp6  & 6.36s & -\\ \hline
\end{tabular}
\caption{Attack Circuit Evaluation Cost}
\label{tab:Cost}
\begin{tabular}{|cc|c|c|c|}
\hline
\multicolumn{2}{|c|}{Probability} & \makecell{Theoretical Attack \\ Success Probability} & \makecell{Experimental Attack \\ Success Probability} & {\# of Iterations}\\ \hline\hline

\multicolumn{1}{|c|}{\multirow{4}{*}{\makecell{\REP{} Attack}}} & Exp1 & 50.00\% & 51.8\% & 1000 \\ \cline{2-5} 
\multicolumn{1}{|c|}{} & Exp2 & 75.40\% & 75.5\% & 1000 \\ \cline{2-5} 
\multicolumn{1}{|c|}{} & Exp3 & 94.34\% & 93.4\% & 1000 \\ \cline{2-5} 
\multicolumn{1}{|c|}{} & Exp4 & 100\% & 100\% & 1 \\
\hline\hline

\multicolumn{1}{|c|}{\multirow{2}{*}{\makecell{$\REP^{\sf msk}$Attack}}} & Exp5 & 1.13\% & 0.8\% &  1000 \\ \cline{2-5} 
\multicolumn{1}{|c|}{} & Exp6 & 2.22\% & 1.8\% & 1000 \\ \hline
\end{tabular}
\caption{Attack success probability (1000 iterations unless otherwise noted)}
\label{tab:Prob}
\end{table}

We remark that the drop in success probability observed in Exp.~5 and Exp.~6 is within one standard deviation.
Also, due to the lack of BFV bootstrapping functionality in the Lattigo library, we implemented the attacks on $\REP$ and $\REP^{\sf msk}$ only for specific parameter settings, and did not implement the attack on $\PE$.  
However in a bootstrappable setting, not only is an attack on $\PE$ feasible, but the above attacks could also be carried out independently of the parameter settings.


\section{Discussion}\label{sec: Discussion}
    \subsection{Homomorphic Cryptanalysis}
    To provide verifiability, the schemes considered in this paper utilize auxiliary secret information in addition to the secret key of the HE scheme: a set of indices $S$ in \REP{} (which also applies essentially to vADDG) and an element $\alpha \in \mathbb{Z}_t^\times$ in \PE{}. Attacks against the schemes described in this paper consist of two steps as follows; 

        \begin{enumerate}
            \item Recover secret values $S$ or $\alpha$ in encrypted state.
            \item Modify a legitimate ciphertext using the encryption of secret values.
        \end{enumerate}

        \subsubsection{Step 1: Recovering Secret Information in an encrypted state.}
        The first step is to recover secret values, $S$ or $\alpha$, in an encrypted state.
        For \REP, we homomorphically find $S$ in two other ways: one way is based on the homomorphic comparison which deterministically recovers $\enc(S)$, and the other way is based on the random characteristic function evaluation, which probabilistically outputs $\enc(S)$ with non-negligible probability.
        
        For \PE, if $\enc(\alpha)$ is recovered we can modify a legitimate ciphertext easily. However, \REQ{} adopted a re-randomized process to make recovering $\enc(\alpha)$ hard. Here the proposed attack does not learn $\alpha$. This is not only because, as mentioned earlier, \PE{} encoding can be performed without knowledge of $\alpha$ or $\enc(\alpha)$, but it can also be explained from the following perspective:
        We can think of the ciphertext polynomial $\mathsf F(Y)\in \cc R^2_q[Y]$ as a ciphertext with two secret keys $$(1,s)\otimes(1,\alpha,\dots,\alpha^{\deg F}), \;\;(1,s)\otimes(1,0,\dots,0).$$ 
        For example, we think of the ciphertext $\ct_0+\ct_1Y$ as an ciphertext with two secret keys $(1,s)\otimes(1,\alpha)=(1,s,\alpha,\alpha s)$ and $(1,s)\otimes(1,0)=(1,s,0,0)$, in the meaning that $(\ct_0,\ct_1)=\big((b_0,a_0),(b_1,a_1)\big)$ would be decrypted by the secret key $(1,s,\alpha,\alpha s)$ or $(1,s,0,0)$. 
        In this sense, 
        $$\big((b_0,a_0),(b_1,a_1)\big)=\big(\enc(0,\dots,0),\enc(1,\dots 1)\big)$$
        is an encryption of $\alpha$ under $(1,s,\alpha,\alpha s)$ and an encryption of $0$ under $(1,s,0,0)$.
        Indeed, one of the possible choice of $\mathsf L^\star(Y)$ is $\enc(0,\dots,0)+\enc(1,\dots, 1)Y$, which is as an encryption of $\alpha$ and $0$ in this very sense.
        
        \subsubsection{Step 2: Modify a legitimate ciphertext}
        The second step is to homomorphically generate a cheating circuit using the information on secret values $\alpha$ or $S$.
        In the plaintext state, the forgery of the secret encoding is just a composition of decoding and encoding, where the decoding and encoding are performed with the aid of the secret values $S$ or $\alpha$, namely $\mathsf{Dcd_{sv}}$ and $\mathsf{Ecd_{sv}}$ where ${\sf sv}$ stands for secret values $S$ or $\alpha$. 
        Thus the attack is just a homomorphic evaluation of such cheating circuits, with the aid of the information of $S$ or $\alpha$ in an  encrypted state: 
        See Fig. \ref{FigForging}.
        
                \begin{figure}[htb!]
                    \centering
                    \begin{tikzpicture}
                      \matrix (m) [matrix of math nodes,row sep=4em,column sep=7em,minimum width=2em]
                      {
                        \nobarfrac{\enc\big(\mathsf{Ecd}(f(v), f(m))\big)}{\enc\big(\mathsf{Ecd}(g(v),g(m))\big)} &\enc\big(\mathsf{Ecd}(f(v),g(m))\big)\\
                        \nobarfrac{\mathsf{Ecd}(f(v), f(m))}{\mathsf{Ecd}(g(v),g(m))}&\mathsf{Ecd}(f(v),g(m)) \\
                        \nobarfrac{f(v),f(m),}{g(v),g(m)} & f(v),g(m)\\
                      };
                      \path[]
                        (m-1-1) edge [-stealth] node[above]{\sf Eval$_{\mathsf{Forge}_\mathsf{sv}}$} (m-1-2)
                                edge [<->] node[left] {$\mathsf{Dec_{sk}}$} node[right] {$\mathsf{Enc_{pk}}$} (m-2-1) 
                        (m-1-2) edge [<->] node[left] {$\mathsf{Dec_{sk}}$} node[right] {$\mathsf{Enc_{pk}}$} (m-2-2)
                        (m-2-1) edge [-stealth] node [above] {\sf Forge$_{\mathsf{sv}}$}(m-2-2)
                                edge [<->] node[left] {$\mathsf{Dcd_{sv}}$} node[right] {$\mathsf{Ecd_{sv}}$} (m-3-1)
                        (m-3-1) edge [-stealth] node[below]{\sf Select} (m-3-2)
                        (m-3-2) edge [<->] node[left] {$\mathsf{Dcd_{sv}}$} node[right] {$\mathsf{Ecd_{sv}}$} (m-2-2)
                        ;
                       \draw[dashed] (-7.2, 0.85) -- (-4, 0.85);
                       \draw[dashed] (-1.5, 0.85) -- (1.4, 0.85);
                       \draw[dashed] (3.65, 0.85) -- (5, 0.85);
                
                       \draw[dashed] (-7.2, -1.15) -- (-4, -1.15);
                       \draw[dashed] (-1.5, -1.15) -- (1.4, -1.15);
                       \draw[dashed] (3.65, -1.15) -- (5, -1.15);
                
                       \node at (-5.5,2.0) {Ciphertext Space};
                       \node at (-5.5,0.35) {Plaintext Space};
                       \node at (-5.5,-1.65) {Message Space};
                    \end{tikzpicture}
                    \caption{An Overview of Homomorphic Forgery on Lightweight VHE. $f$ and $g$ denote the requested circuit and a malicious circuit, respectively.} \label{FigForging}
                \end{figure}
                
        For \REP{}, the secret value for Decoding and Encoding is the verification slot $S$. Accordingly, the homomorphic decoding is just masking a ciphertext by multiplying the encryption of $J_S$ and $J_{S^c}$, and the homomorphic encoding is adding ciphertexts masked by $J_S$ and $J_{S^c}$.

        For \PE{}, the decoding involves evaluating at \(Y=\alpha\) and \(Y=0\), while the encoding process interpolates \(\mathbf{m}, \mathbf{v} \in \mathbb{Z}_t^N\) at \(Y=0\) and \(Y=\alpha\) to obtain
        \[
        \mathbf{m} + (\mathbf{v}-\mathbf{m})\alpha^{-1}\cdot Y.
        \]
        According to the earlier perspective that \( {\sf L}^\star := \enc(0) + \enc(1)Y \) can be seen as an encryption of \(0\) and \(\alpha\), interpolating the two ciphertext polynomials \({\sf F}\) and \({\sf G}\) yields
        \begin{align*}   
        &{\sf G}+({\sf F}-{\sf G})\cdot ({\sf L}^\star)^{-1}Y \\
        = &{\sf G}+({\sf F}-{\sf G})\cdot ({\sf L}^\star)^{\phi(t)-1}Y.
        \end{align*}
        This exhibits the structure of the previously discussed cheating circuit.

\section{Conclusion}
In this paper, we presented attacks against the verifiable homomorphic encryption schemes introduced in \cite{VE, cryptoeprint:PRF}.

For the vADDG scheme of \cite{cryptoeprint:PRF}, we introduced a shallow pseudorandom characteristic function to overcome the limitation imposed by bootstrapping depth. The attack succeeds with probability $\left(1-(1/2)^\alpha\right)\cdot (1/2)^\beta$ and has computational cost $O(\gamma)$.  
For \REP{} in \cite{VE}, we constructed a circuit that homomorphically computes the position vector of the common values. This attack has cost $O(k\log t)$ and succeeds with probability 1 when $k=n/2+1$.  
Our patching solution using multiple secret keys reduces the success probability to \(O(n^{-1})\), which yields only very weak security, arguably sufficient only against a covert adversary.  
For \PE{}, we exploit Euler's theorem to construct a deterministic forgery attack with computational cost $O(\log t)$.

One possible approach to decreasing the success probability of attacks is to restrict the circuit depth that can be homomorphically evaluated. However, as demonstrated by the attack on vADDG, a cheating circuit can still be constructed even at shallow depths. Moreover, an adversary may perform forgery despite risking computation failure due to noise robustness. Hence, the approach of limiting depth requires more thorough analysis.

Another approach is to employ multiple secrets, as discussed in Section~\ref{subsec: diff keys}.  
This restricts the class of HE circuits that can be evaluated.  
However, as shown in Section~\ref{subsec: diff keys}, simply assigning different secrets slotwise is insufficient to achieve malicious security.  
Thus, this approach of using multiple secrets also requires further and more precise analysis.


\section*{Acknowledgement}
We would like to thank Damien Stehl{\'e} for his guidance and valuable comments on this research.  
We also thank Changmin Lee and the anonymous Asiacrypt2025 reviewers for their helpful feedback.
This research was supported by the Institute of Information and Communications Technology Planning and Evaluation (IITP), grant funded by the Korea Government (MSIT) (RS-2021-II210727, A Study on Cryptographic Primitives for SNARK)


\bibliographystyle{alpha}
\bibliography{biblo}


\newpage
\appendix
\section{BFV Scheme}\label{Appendix BFV}

We recall the details of the BFV homomorphic encryption scheme~\cite{Bra, bfv}. 
For more comprehensive descriptions, we refer the reader to~\cite{BFVHElib}.

\subsubsection{Structure of $\cc R_t$.}
Let $M$ be a power of $2$ and define the cyclotomic ring 
$\mathcal R = \mathbb Z[X]/(\Phi_M(X))$ and $\mathcal R_q=\mathcal R/q\mathcal R$ for $q\in \mathbb N$.
Let $p$ be a prime and let $d$ be the multiplicative order of $p$ in $\mathbb Z_M^\times$.
For $t=p^r$, there exists an injective ring homomorphism $\mathbb Z_t^N\hookrightarrow \mathcal R_t$ where $N=M/d$.
Moreover, the Galois automorphisms on $\mathcal R_t$ corresponds to some permutation actions, called a `rotation’, on $\mathbb Z_t^N$.
Thus, operations on slots (i.e., on $\mathbb Z_t^N$) are realized via the ring operations in $\mathcal R_t$.

When encoding data from $\mathbb Z_t$ into $\mathcal R_t$, one may either use the injective ring homomorphism described above or directly encode values into the coefficients of a polynomial in $\mathcal R_t$. 
The former is referred to as \emph{slot encoding}, while the latter is called \emph{coefficient encoding}. 
Note that coefficient encoding is not ring-homomorphic, but it is occasionally used in special cases. 

\subsubsection{BFV Scheme.} The BFV scheme is a homomorphic encryption scheme that supports both ring operations and Galois automorphisms over $\mathcal R_t$. 
Its ciphertext space is $\mathcal R_q^2$ for $q>t$.
Let $\Delta=\lfloor q/t\rfloor$, and let $\chi_{\mathsf{key}}$ and $\chi_{\mathsf{err}}$ be probability distributions over $\mathcal R_q$.
\begin{itemize}
    \item \textsf{KeyGen}: Sample a secret key $s \gets \chi_{\mathsf{key}}$ and generate the associated evaluation keys. The relinearization key $\rlk$ is generated so that, for any $a \in \mathcal R_q$, one can (oblivious to $s$) find $\alpha,\beta \in \mathcal R_q$ such that 
    \([as^2]_q \;\cong\; [\alpha s + \beta]_q.\)
    Similarly, for a Galois automorphism $\sigma$, the rotation key $\rtk_\sigma$ is generated so that, 
    for any $a \in \mathcal R_q$, one can (oblivious to $s$) find $\alpha,\beta \in \mathcal R_q$ such that 
    \([a \sigma(s)]_q \cong [\alpha s + \beta]_q.\)

    \item \textsf{Enc}: For a message $m\in \mathcal R_t$ and secret key $s$, sample $c_1\gets U(\mathcal R_q)$ and $e\gets \chi_{\mathsf{err}}$, and output  
    \(
    (c_1,c_0) = (c_1,\; -c_1s+\Delta m +e) \in \mathcal R_q^2.
    \)
    \item \textsf{Dec}: For a ciphertext $(c_1,c_0)$ and secret key $s$, output  
    \(
    \left\lfloor \tfrac{t}{q}[c_1s+c_0]_q \right\rceil \in \mathcal R_t.
    \)
    \item \textsf{Eval}: For simplicity, here we identify a plaintext $m$ with an insecure ciphertext $(0,\Delta m)$, so that plaintext-ciphertext operations can be described as ciphertext–ciphertext operations.  
        \begin{itemize}
            \item \textsf{Add}: Given $(c_1,c_0)$ and $(c_1',c_0')$, output  
            \(
            (c_1+c_1', c_0+c_0').
            \)
            \item \textsf{Mult}: Given $(c_1,c_0)$ and $(c_1',c_0')$, first compute  
            \(
            (d_2,d_1,d_0) = (c_1c_1', c_1c_0'+c_1c_0', c_0c_0').
            \)  
            Scale each by $t/q$, round, and mod-out $q$, i.e., $([\lfloor \tfrac{t}{q}d_2\rceil]_q,$ $[\lfloor \tfrac{t}{q}d_1\rceil]_q, [\lfloor \tfrac{t}{q}d_0\rceil]_q)$.  
            Then, using $\rlk$, convert $[\lfloor \tfrac{t}{q}d_2\rceil]_q\in \cc R_q$ into $(\alpha,\beta)$ such that $[\lfloor \tfrac{t}{q}d_2\rceil s^2]_q \cong [\alpha s+\beta]_q$.  
            Output  
            \(
            \bigl([\lfloor \tfrac{t}{q}d_1\rceil+\alpha]_q,\; [\lfloor \tfrac{t}{q}d_0\rceil+\beta \bigr]_q).
            \)
            \item \textsf{Rot$_{\sigma}$}: Given $(c_1,c_0)$ and a Galois automorphism $\sigma$, compute $(\sigma(c_1),\sigma(c_0))$.  
            Using $\rtk_{\sigma}$, convert $\sigma(c_1)$ into $(\alpha,\beta)$ such that $[\sigma(c_1)\sigma(s)]_q \cong [\alpha s+\beta]_q$.  
            Output
            \(
            (\alpha,\; [\sigma(c_0)+\beta]_q).
            \)
        \end{itemize}

\end{itemize}

    \paragraph{Modulus Switching.} One can change the ciphertext modulus $q$. 
    For $(c_1,c_0)\in \mathcal R_q^2$, the pair
    \(
        \bigl(\,\lfloor \tfrac{q'}{q}c_1 \rceil,\; \lfloor \tfrac{q'}{q}c_0 \rceil \,\bigr) \in \mathcal R_{q'}^2
    \)
    is an encryption of the same plaintext, and this procedure is called modulus switching.

    \paragraph{Slot–Coefficient Transformation.}
    When computations on the coefficients themselves are required, it is also possible to switch between slot encoding and coefficient encoding in the encrypted state. 
    Such conversions essentially amount to $\mathbb Z_t$-linear algebra over $\mathcal R_t$ 
    and can be performed using only rotation keys and the relinearization key.
    
\section{Attacks on $\REP$ and $\REP^\mathsf{msk}$ for $t=p^r$}\label{Appendix A}

The attack circuit described in the main text is a function over $\mathbb{Z}_t$.  
When $t=p$ is prime, every function can be represented by a polynomial.  
However, when $t$ is not prime, not every function $f:\mathbb{Z}_t \to \mathbb{Z}_t$ can be represented by a polynomial over $\mathbb{Z}_t$, since in this case $p$ is not invertible in $\mathbb{Z}_t$.
Depending on the underlying HE scheme, however, there exists a way to overcome this issue.  
This appendix provides additional technical details for extending our attacks to the case $t=p^r$ on the BFV scheme. For BGV, one can apply the standard BFV-BGV ciphertext conversion 
(cf. Appendix~A in~\cite{bfv-bgvconversion}).

\subsection{Evaluating Non-polynomial Functions.}
\subsubsection{Divide-by-$p$ and Its Inverse.}
One can evaluate the `divide-by-$p$' function
\[
    p\mathbb{Z}_{p^r} \;\overset{p^{-1}}{\longrightarrow}\; \mathbb{Z}_{p^{r-1}}
\]
together with its inverse operation.
Consider a BFV ciphertext $(c_0,c_1)\in \cc R^2_q$ satisfying $c_0+c_1s = \Delta p m + e + qI$, where $pm$ is the plaintext of multiple of $p$ and $e$ is error. If we interpret the scaling factor as $\Delta' = \Delta p$,  
then the plaintext modulus is naturally $p^{r-1}$,  
and the ciphertext can be regarded as an encryption of $m$ under this modulus.

\subsubsection{Divide-by-$p^{r-1}$ without Modulus Reduction.}
However, while the above operation is computationally free, it has the drawback of reducing the plaintext modulus, which in turn limits the available $p$-adic precision. We now show that it is possible to implement the following operation for a BFV ciphertext without reducing the plaintext modulus:
\[
p^{r-1}\cdot \mathbb{Z}_{p^r} \xrightarrow{(p^{r-1})^{-1}} \mathbb{Z}_{p^r}.
\]

First, we modulus-switch the BFV ciphertext so that $q$ becomes a multiple of $t=p^r$ ($q=\Delta \cdot p^r$). 
Then, assuming via a slot–coefficient transformation that the message is coefficient-encoded, 
the BFV ciphertext $(c_0, c_1)\in \mathcal{R}^2_q$ satisfies
\(
    c_0 + c_1 s = qI + \Delta p^{r-1} m+ e.
\)
Then,
\(
\left(\left\lfloor \tfrac{c_0}{p^{r-1}}\right\rceil,\; \left\lfloor \tfrac{c_1}{p^{r-1}}\right\rceil\right) 
   \in \mathcal{R}^2_{q/p^{r-1}}
\)
satisfies
\begin{align*}
\left\lfloor \tfrac{c_0}{p^{r-1}} \right\rceil + \left\lfloor \tfrac{c_1}{p^{r-1}} \right\rceil s
&= \tfrac{q}{p^{r-1}} I + \Delta m 
   + \left\lfloor \tfrac{e}{p^{r-1}} \right\rceil 
   + (\delta_0+\delta_1 s+\delta_2)\\
&= \Delta (pI+m) + \left\lfloor \tfrac{e}{p^{r-1}} \right\rceil 
   + (\delta_0+\delta_1 s+\delta_2)
\end{align*}
where the $\delta_i$ are small polynomials arising from the rounding, of size less than $1$.
Embedding this ciphertext back into $\mathcal{R}^2_q$, 
it becomes an encryption of 
\(
(pI+m) \bmod p^r.
\)
By subsequently applying the digit-extraction polynomial~\cite{digitextraction} on the coefficients, one can recover an encryption of $m \bmod p^r$. This operation is computationally expensive, involving slot-coefficient trasformations, yet it can be evaluated in polynomial time.

\subsection{Attacks on $\REP$}
We first modify the $\sf Compare$ function by evaluating its multiple $\phi(t)\cdot \sf Compare$.  
Note that $\phi(t)=\phi(p^r)=(p-1)p^{r-1}.$

\begin{lemma}
$\mathsf{Eval}_{\phi(t)\mathsf{Compare}}$ can be evaluated with $O(\log t)$ homomorphic multiplications with plaintext modulus $t=p^r$ for $r>1$.
\end{lemma}

\begin{proof}
Let $e(x)=1-x^{\phi(t)}$, and define
\[
    \alpha(x):=e(x) \cdot \prod_{u\in \mathbb{Z}_t^\times\setminus \{1\}}(x+1-u).
\]
Note that
\[
\prod_{u\in \mathbb{Z}_t^\times\setminus \{1\}}(x+1-u)=\frac{(x+1)^{\phi(t)}-1}{(x+1)-1}=\sum_{i=0}^{\phi(t)-1}(x+1)^i,
\]
so $\alpha$ can be evaluated with $O(\log t)$ homomorphic multiplications.  
Now consider the following cases:
\begin{enumerate}
    \item If $x=0$, then $\alpha(0)=e(0)\cdot \phi(t)=\phi(t)$. 
    \item If $x\in \mathbb{Z}_t^\times$, then $e(x)=0$ and hence $\alpha(x)=0$.
    \item If $x$ is nonzero and non-unit, then $x+1\in \mathbb{Z}_t^\times\setminus\{1\}$, so  
    $\alpha(x)=e(x)\cdot\prod_{u}(x+1-u)=e(x)\cdot 0=0.$
\end{enumerate}
Thus $\alpha(x-y)=\phi(t)\cdot \textsf{Compare}(x,y)$.      
\end{proof}

By replacing $\mathsf{Eval}_{\sf Compare}$ with $\mathsf{Eval}_{\phi(t)\cdot\sf Compare}$ in Algorithm~\ref{alg: common value position}, we obtain:
\begin{itemize}
    \item For $i \in S^c$: $\ct_{{\sf CVS},i}$ is an encryption of $\phi(t)\cdot J_{S^c}$.
    \item For $i \in S$: $\ct_{{\sf CVS},i}$ is an encryption of $\phi(t)\cdot e_i$, where $e_i$ denotes the $i$th standard vector.
\end{itemize}

Consequently, the resulting ciphertext $\eval_f(\ct)+\ct_{{\sf CVS},i}$ constitutes a forgery with probability $1/2$.

\subsubsection{Increasing Attack Success Probability.}
We now modify Algorithm~3 to increase the attack success probability.  
Let $\ct_{{\sf bool}, i}$ denote the output obtained by applying the function $x \mapsto x-\phi(t)$ coefficient-wise to ${\sf RotSum}(\ct_{{\sf CVS},i})$.  
Then:
\begin{itemize}
\item For $i \in S^c$: $\ct_{{\sf bool}, i}$ is an encryption of $((\tfrac{n}{2}-1)\cdot\phi(t),\dots,(\tfrac{n}{2}-1)\cdot\phi(t))$.
\item For $i \in S$: $\ct_{{\sf bool}, i}$ is an encryption of $(0,\dots,0)$.
\end{itemize}

If $(\tfrac{n}{2}-1) \mid p$, then both cases degenerate to the zero vector, which causes a problem.  
For such exceptional parameter choices, one can resolve the issue by multiplying $(1,\dots,1)-e_j$ to $\ct_{{\sf CVS},i}$ for some $j\neq i$, and then applying the above procedure.  
This yields $(\tfrac{n}{2}-2)\cdot\phi(t)$ for $i\in S^c$ instead, thereby avoiding the problem.  
In the following, we assume without loss of generality that $(\tfrac{n}{2}-1) \nmid p$.

When attempting to multiply the resulting ciphertext with $\ct_{{\sf CVS}, i}$ as in Algorithm~3, an issue arises.  
Since $\ct_{{\sf CVS}, i}$ encrypts multiples of $\phi(t)$ in its slots, multiplying it with the output $\ct_{{\sf bool}, i}$ reduces the plaintext to $0$. We can apply the previously described `divide-by-$p$' technique here.  
Since both $\ct_{{\sf bool}, i}$ and $\ct_{{\sf CVS}, i}$ contain plaintexts that are multiples of $p^{r-1}$, we first divide them by $p^{r-1}$ while setting the plaintext space to $\mathbb{Z}_p$.  

As a result, let $\ct_{{\sf mult}, i}$ denote the ciphertext obtained by this procedure:  
\begin{itemize}
\item For $i \in S^c$: $\ct_{{\sf mult}, i}$ is an encryption of $(\tfrac{n}{2}-1)(p-1)^2\cdot J_{S^c} = (\tfrac{n}{2}-1)\cdot J_{S^c}$.  
\item For $i \in S$: $\ct_{{\sf mult}, i}$ is an encryption of the zero vector.  
\end{itemize}

We can further increase the forgery success probability by summing $\ct_{{\sf mult}, i}$ over $i$ as in Theorem~\ref{lem: attack on rep}. \\

However, if the $p$ is small enough, an edge case arises when the number of indices $i \in S^c$ is a multiple of $p$.  
Since the summation is performed over $\mathbb{Z}_p$, this issue can be avoided when $p \neq 2$ by applying normalization after each addition.  
Concretely, one can apply a polynomial $q\in\bb Z_p[X]$ such that $q(0)=0$ and $q(\tfrac{n}{2}-1)=q(2(\tfrac{n}{2}-1))=\tfrac{n}{2}-1$ so that encryptions of $(\tfrac{n}{2}-1)J_{S^c}$ remains unchanged under repeated additions. 

For the case $p=2$, it is still possible to perform the `divide-by-$p^{r-1}$' 
operation without reducing the plaintext modulus. As noted above, an encryption of $J_{S^c} \pmod{2}$ can be regarded 
as an encryption of $2^{r-1}\cdot J_{S^c} \pmod{2^r}$. 
By dividing by $2^{r-1}$, we obtain an encryption of $J_{S^c} \pmod{2^r}$. 
For typical parameters (e.g., $r \ge 16$ and $n \leq128$), the summation of $J_{S^c} \pmod{2^r}$ does not exceed $2^r$.

\subsection{Attacks on $\REP^{\sf msk}$}
The only difference from the main text lies in the evaluation of the pseudorandom characteristic function in Lemma~\ref{bfvprf}.  
As in the main text, let $U_{t,d} = \{x: x^{\phi(t)/d}=1\}$.

\begin{lemma} For $t=p^r$ and $d \mid (p-1)$, we have
    \[
        d\cdot \chi_{U_{t,d}}(x)=x^{\phi(t)}\cdot\sum_{i=0}^{d-1} \left(x^{\phi(t)/d}\right)^i.
    \]
\end{lemma}

Before the proof, note that $a\in \mathbb{Z}_{p^r}$ is a zero divisor if and only if $a \equiv 0\pmod p$.

\begin{proof}
\begin{enumerate}
    \item Suppose $x$ is a non-unit. Then the left-hand side (LHS) is zero.  
    Since $x$ is a multiple of $p$, the right-hand side (RHS) is also zero, as it becomes a multiple of $p^{\phi(t)}\equiv0 \pmod{p^r}$ ($r\leq (p-1)p^{r-1}=\phi(t)$)
    \item Suppose $x\in U_{t,d}$. Since $x$ is a unit, $x^{\phi(t)}=1$. Also $x^{\phi(t)/d}=1$ by definition of $U_{t,d}$. Thus both LHS and RHS equal $d$.
    \item Suppose $x\in \mathbb{Z}_t^\times\setminus U_{t,d}$. Since $x$ is a unit, $x^{\phi(t)}=1$. Moreover $x^{\phi(t)/d}\neq 1$ by definition, while
    \[
        (x^{\phi(t)/d})^d-1=(x^{\phi(t)/d}-1)\left(\sum_{i=0}^{d-1} (x^{\phi(t)/d})^i\right)=0.
    \]
    Hence $\sum_{i=0}^{d-1} (x^{\phi(t)/d})^i=0$, unless both $x^{\phi(t)/d}-1$ and the sum are zero divisors.  
    To rule out this second case, let $y=x^{\phi(t)/d}$. If $y-1$ is a zero divisor, then $y=pm+1$ for some integer $m$, which implies $\sum_{i=0}^{d-1} y^i\equiv d\not\equiv0 \pmod p$.  
    Thus it is not a zero divisor, a contradiction. Hence RHS $=0$.
\end{enumerate} 
\end{proof}

This completes the extension of Lemma~\ref{bfvprf} to the case $t=p^r$. 
As noted in the main text, one can further combine such characteristic functions to efficiently construct another characteristic function whose support size approximates a desired value.
In the case $p=2$, however, the lemma does not yield a meaningful result. 
We leave the study of efficient pseudorandom characteristic function evaluation for this case as future work.\\

We note that the cryptanalysis on $\PE$ in Section~\ref{sec:PE} naturally extends to the case $t=p^r$, 
and thus requires no further explanation.

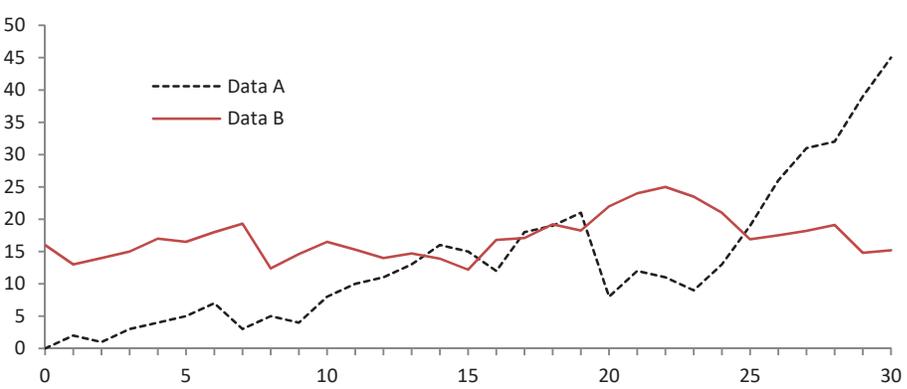
\end{document}